\documentclass[conference,a4paper,9pt]{IEEEtran}
\usepackage[utf8]{inputenc}
\usepackage{amssymb}
\usepackage{amsthm}
\usepackage{mathtools}
\usepackage{xcolor}
\usepackage{bm}
\usepackage{makecell}
\usepackage{booktabs}
\usepackage{subfig}
\usepackage{subcaption}
\usepackage[font=footnotesize]{caption}

\newtheorem{theorem}{Theorem}     
  
\newtheorem{lemma}{Lemma}         
\newtheorem{prop}{Proposition}    

\newcommand{\ok}{{\rm ok}}
\newcommand{\ko}{{\rm ko}}
\newcommand{\tails}{{\rm tails}}
\newcommand{\unknown}{{\rm unk}}
\newcommand{\uniform}{{\rm unif}}

\title{Goal-Oriented Joint Source-Channel Coding:\\ Distortion - Classification - Power Trade-off}
 \author{
 	\IEEEauthorblockN{
 		Andriy Enttsel \IEEEauthorrefmark{1}, 
 		Weichen Wang \IEEEauthorrefmark{3}, 
 		Mauro Mangia \IEEEauthorrefmark{1}\IEEEauthorrefmark{2}, 
 		Riccardo Rovatti \IEEEauthorrefmark{1}\IEEEauthorrefmark{2},
 		Deniz Gündüz \IEEEauthorrefmark{3}}
 		\IEEEauthorblockA{\IEEEauthorrefmark{1}DEI, 
        \IEEEauthorrefmark{2}ARCES, University of Bologna, Italy,
        \IEEEauthorrefmark{3}Dept. of Electrical and Electronic Eng, Imperial College London, UK}
 		\IEEEauthorblockA{\{andriy.enttsel,~mauro.mangia,~riccardo.rovatti\}@unibo.it - \{weichen.wang18,~d.gunduz\}@imperial.ac.uk}
}

\makeatletter
\g@addto@macro\normalsize{%
  \setlength\abovedisplayskip{5pt}
  \setlength\belowdisplayskip{5pt}
  \setlength\abovedisplayshortskip{5pt}
  \setlength\belowdisplayshortskip{5pt}
}
\makeatother

\begin{document}

\maketitle
\begin{abstract}
Joint source-channel coding is a compelling paradigm when low-latency and low-complexity communication is required. This work proposes a theoretical framework that integrates classification and anomaly detection within the conventional signal reconstruction objective. Assuming a Gaussian scalar source and constraining the encoder to piecewise linear mappings, we derive tractable design rules and explicitly characterize the trade-offs between distortion, classification error, and transmission power.
\end{abstract}
\begin{IEEEkeywords}
Joint source-channel coding, classification, anomaly detection, semantic communication.
\end{IEEEkeywords}
\section{Introduction}
\label{sec:intro}

In practical applications where low latency and low complexity are critical, joint source channel coding (JSCC) offers an attractive alternative to classical digital schemes. By combining source compression and channel coding into a single operation, JSCC eliminates the need for buffering and block-based processing, typical of digital codecs, and enables “zero delay” transmission \cite{Gündüz_ProcIEEE2024}. 

However, the main limitation of JSCC lies in the lack of a systematic design procedure for encoder and decoder mappings covering general source-channel pairs. Early work derived optimal linear mappings \cite{Goblick_TIT1965,Lee_1976TC}. Later results highlighted their suboptimality when there is a mismatch between the source and channel bandwidths, as well as provided necessary conditions for optimal nonlinear encoder and decoder mappings under a power constraint \cite{Akyol_2014TIT}. Yet, the numerical methods used to find these mappings are often sensitive to local minima.

In response, more recent research has developed deep learning approaches to JSCC, with promising results on complex sources such as images and text \cite{Farsad_2018ICASSP, Bourtsoulatze_TCCN2019, Dai_JSAC2022}. Nevertheless, closed-form optimal mappings remain known only for specific cases, such as scalar Gaussian sources over additive white Gaussian noise (AWGN) channels.
A comprehensive review on JSCC can be found in \cite{Gündüz_ProcIEEE2024}.

A complementary trend is the design of semantic and task-oriented communication systems \cite{Kalfa_2021DSP, Gündüz_JSAC2023}. These paradigms reduce communication overhead by transmitting only information relevant to accomplish a downstream task, such as classification, rather than reconstructing the full signal. This is particularly beneficial for machine-oriented communication scenarios, where the objective is to process high-dimensional data with minimal delay and bandwidth consumption \cite{Gündüz_JSAC2023}.

One of the most common downstream tasks is classification. A substantial body of work has examined classification-aware lossy compression. Earlier studies \cite{Oehler_1995TPAMI, Baras_1999TIT, Gupta_2003TIT} focused on tailoring vector quantizers by integrating classification-related terms in the distortion measure to explicitly manage the distortion–classification trade-off. More recent work integrated classification in neural network-based codecs \cite{Luo_DCC2021} and JPEG compression standard \cite{Bai_2022AAAI}.

\begin{figure}

\centering
\includegraphics[width=\columnwidth]{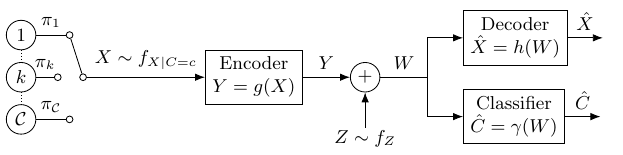}
\caption{JSCC scheme for a source corrupted by additive noise to be reconstructed and classified at the receiver.}
\label{fig: scheme}
\end{figure}

Anomaly detection (AD) is also gaining particular attention \cite{Marchioni_TSMC2024, Enttsel_2024AICAS, Enttsel_EUSIPCO2024}. Unlike conventional classification, which assumes that all classes are known a priori, AD involves only a single nominal class representative of normal behavior. For this reason, it is often referred to as one-class classification. The goal is to identify deviations from this class, including all events that are rare, unexpected, or qualitatively different, which may compromise system operation \cite{Ruff_2021ProcIEEE}. Therefore, AD naturally requires low-latency, low-bandwidth communication, making zero-delay analog schemes especially well suited. The work in \cite{Marchioni_TSMC2024} analyzes how lossy digital schemes optimized for sole reconstruction affect the subsequent anomaly detection task. In \cite{Enttsel_2024AICAS, Enttsel_EUSIPCO2024}, the authors design autoencoder-based schemes that manage the trade-off between anomaly detection and reconstruction.

This work aims to extend the theoretical results of JSCC beyond the objective of classical signal reconstruction by incorporating classification and anomaly detection within the system design. 

In the case of binary classification, we assume a Gaussian source with balanced binary classes and constrain the encoder to a piecewise linear form. We also report learning-based numerical results without constraints, which provides greater performance at the cost of computational complexity. 

For anomaly detection, we show that, under similar assumptions, a modified JSCC scheme can be designed to simultaneously preserve signal fidelity under normal conditions and transmit sufficient statistics to enable anomaly detection at the receiver. Proofs of the main theoretical results are reported in Appendices A and B.

\section{Problem setup}
\label{sec:problem setup}

We consider a source pair $(X,C)$, where $X \in \mathbb{R}$ is modeled as a mixture of $ \mathcal{C} $ distributions, $C$ is the discrete class: $C \in \mathbb{Z}^+_{< \mathcal{C}}$, so that $f_X(x) = \sum_{c=0}^{\mathcal{C}-1} \pi_c f_{X|C=c}(x)$, where $\pi_c$ is the weight of each class $c$. The source is to be transmitted through a single AWGN channel with zero delay. The objective is to reconstruct $X$ as well as obtain the classification $C$ at the receiver. As shown in Fig.~\ref{fig: scheme}, we consider a JSCC scheme with a separate decoder $h$ and a classifier $\gamma$ to achieve optimal information utilization. The optimal encoder-decoder-classifier triple $\left(g^\star, h^\star, \gamma^\star \right)$ can then be defined as the solution to the following distortion-classification-power (DCP) optimization problem. 
\begin{align}
\label{eq:DCP}
\tag{DCP}
(g^\star, h^\star, \gamma^\star)
&= \arg\min_{\substack{g \in \mathcal{G},\, h \in \mathcal{H},\, \gamma \in \Gamma}}
   \mathrm{MSE}\!\left(X,\hat{X}\right) \\
   \notag
\text{s.t.}\quad
& \Pr \{C \neq \hat{C}\} \le P_e, \quad \mathbb{E} \left[Y^2\right] \le P,
\end{align}
where $\mathcal{G},\mathcal{H} \subseteq \mathcal{M}(\mathbb{R},\mathbb{R}) $ are the measurable deterministic\footnote{Determinism can be derived following the same arguments as in \cite{Akyol_2014TIT}.} function spaces and $\Gamma$ is the set of all deterministic\footnote{We make the assumption that the measure of set of events with equal posteriors is zero, which is realistic. Then determinism follows from the definition of Bayes classifier.} functions $\gamma: \mathbb{R} \to \mathbb{Z}^+_{< \mathcal{C}}$. That is, $( g , h, \gamma)$ are optimized within $ (\mathcal{G},\mathcal{H}, \Gamma)$ to minimize the mean squared error (MSE), subject to a constraint in classification error $P_e$ and the transmitted average power $P$. Note that while $(X,C)$ may have a complicated dependency, in this work, we concentrate on the case where $c(x)$ is well-defined, i.e., $C$ can be uniquely determined by the value of $X$.

\section{Binary classification}
\label{sec: binary class}
\subsection{Preliminaries}
To solve DCP analytically, without loss of generality, we consider a zero-mean\footnote{This is easily reducible from an arbitrary Gaussian, with constant shifting.} Gaussian source $X \sim \mathcal{N}(0, \sigma_X^2)$ with balanced disjoint binary classes, i.e., $C = 1\left\{ x>0 \right\}$  so that $f_X(x) = 0.5f_{X|C=0}(x) + 0.5  f_{X|C=1}(x)$.
As channel noise, we consider additive Gaussian white noise (AWGN): $Z \sim \mathcal{N}(0, \sigma_Z^2)$. 

We limit the set of encoding functions $\mathcal{G}$ to piecewise linear mappings of the form:
\begin{equation}
\label{eq: encoder}
\mathcal{\bar{G}} = \left\{ g(x) = A x + B \cdot \mathbf{1}_I(x) \;:\; A, B \in \mathbb{R}_+ \right\},
\end{equation}
where $\mathbf{1}_I(\cdot)$ is an indicator function with support of finite unions of open intervals. It can be immediately seen that under balanced binary class, due to symmetry, the indicator function reduces to Heaviside function, which is $\operatorname{sign}(x)$.

We may now see that $\mathcal{\bar{G}}$, paired with the power constraint, can also be viewed as a linear combination of the optimal encoders for two different objectives: reconstruction and classification.
For the former, the optimal encoding is given by $ g(x) =  x \sqrt{P}/ \sigma_X $, while for the latter, the problem can be viewed as analog bit transmission, where the optimal encoding that maximizes classification performance is $ g(x) = \sqrt{P} \, \operatorname{sign}(x) $.

We work with dimensionless parameters: $\alpha = A\sigma_X / \sigma_Z$, $\beta = B / \sigma_Z$ and $\operatorname{SNR}=P / \sigma_Z^2$ and normalized variables $\tilde{x} = x / \sigma_X$, 
$\tilde{z} = z / \sigma_Z, \tilde{w} = w /\sigma_Z
= \tilde g (\tilde x) + \tilde{z} = \alpha\,\tilde{x} + \beta\,\operatorname{sign}(\tilde{x}) + \tilde{z}$ and 
$\tilde h(\tilde w) = \hat{x} / \sigma_X$.
All subsequent expressions are written in terms of the standard normal probability density function (PDF), and tail function
\begin{equation}
\phi(\xi) =  \exp\left( -\xi^2 / 2 \right) / \sqrt{2\pi},\qquad
Q(\xi) = \textstyle{\int_{\xi}^{\infty}} \phi(s)\,d s.
\end{equation}
We also rely on Owen’s T function $T(h,a)$ \cite{Owen_1956AMS}, which yields compact
expressions for the cumulative distribution function $\Phi_{\mathrm{SN}}$ and the PDF $\phi_{\mathrm{SN}}$ of the skew-normal distribution with shape parameter $\lambda$ \cite{Azzalini_1985SJoS}:
\begin{align}
\Phi_\mathrm{SN}(\xi;\lambda) &= Q(- \xi) - 2\,T(\xi,\lambda),\\
\label{eq: skew pdf}
\phi_\mathrm{SN}(\xi;\lambda)& = 2\,\phi(\xi)\,Q(-\lambda \xi).
\end{align}

\subsection{Main results}
The first main result is the decoder's closed-form expression for the considered class of encoders. Proofs of the results presented in this subsection are provided in Appendix~A.
\begin{lemma}[Decoder]
    \label{lemma: decoder}
     For an arbitrary encoder $g \in \mathcal{G}$ the minimum mean square error (MMSE) estimator decoder is
     \begin{align}
        &\tilde h_{\alpha,\beta}(\tilde w)=
        \frac{2}{\pi\sqrt{\alpha^2+1}}\,
        \frac{\exp\left[- \left(\tilde w^2 + \beta^2 \right) / 2\right] \sinh\left(\tilde w \beta \right)}{S_-\left(\tilde w \right)+S_+\left(\tilde w \right)}  \\
        &+\frac{\alpha}{\alpha^2+1}\,
        \frac{(\tilde w+\beta)\,S_-(\tilde w)+(\tilde w-\beta)\,S_+(\tilde w)}
        {S_-(\tilde w)+S_+(\tilde w)} \notag 
        \end{align}
    where $S_{\pm}(\tilde w)=\phi_{\mathrm{SN}}\left[(\tilde w\mp\beta )/\sqrt{\alpha^2+1} ;\pm\alpha\right]$.
\end{lemma}

The two notable corner cases are: \textit{i}) $\beta = 0$ (linear encoding, optimal for reconstruction) resulting in the well-known optimal decoder $h_{\alpha, 0}(\tilde{w}) = \alpha \tilde{w} / (1 + \alpha^2)$; \textit{ii}) $\alpha = 0$ (BPSK modulation, optimal for discrete information transmission) leading to $h_{0, \beta}(\tilde{w}) = \sqrt{2 / \pi} \tanh \left( \beta \tilde{w} \right)$. 

Lemma~\ref{lemma: decoder} allows us to compute the MMSE. Although the general expression is not tractable analytically, it is possible to derive MMSE for the two corner cases $\beta = 0$ and $\alpha = 0$:
\begin{align}
\label{eq: mmse linear}
   &\mathrm{MMSE}|_{\beta=0}= \left(1 +\alpha^2 \right) ^{-1}, \\
\label{eq: mmse nonlinear}
    &\mathrm{MMSE}|_{\alpha=0} \lesssim \frac{\pi-2}{\pi} 
    + \frac{8}{\pi}\Phi_{\mathrm{SN}}\!\left(
\frac{-\beta^2\sqrt{\pi}}{\sqrt{2+\pi\beta^2}};\;
\frac{1}{\sqrt{1+\pi\beta^2}}
\right). 
\end{align}
The expression in \eqref{eq: mmse linear} is a well-known result, while \eqref{eq: mmse nonlinear} represents a tight upper bound to the true MMSE, obtained by replacing the optimal decoder with
an $\operatorname{erf}$-based surrogate exploiting the approximation $\tanh(\xi) \simeq \operatorname{erf}\left( \xi \sqrt{\pi} / 2 \right)$, with $\operatorname{erf}(\xi) \triangleq 2Q(-\xi\sqrt{2}) - 1$ denoting the error function.

The optimal detector and its performance in terms of the risk $ \mathcal{R}( \gamma ) \triangleq \Pr \{C \neq \hat{C}\}$ are characterized by the following lemma.
\begin{lemma}[Classifier]
    \label{lemma: classifier}
    With the encoder $g \in \mathcal{G}$, the optimal Bayesian binary classifier is $\gamma_{\alpha,\beta}(\tilde w) = \mathbf{1}\left\{ \tilde w > 0 \right\}$ characterized by the risk
    \begin{equation}
        \label{eq: risk classifier}
        \mathcal{R}\left( \gamma_{\alpha, \beta} \right) = \Phi_\mathrm{SN}\left( - \beta / \sqrt{\alpha^2 + 1};  \; \alpha \right).
    \end{equation}
\end{lemma}

Again, for the two corner cases $\beta = 0$ and $\alpha = 0$, \eqref{eq: risk classifier} yields $\mathcal{R}( \gamma_{\alpha, 0}) = \operatorname{arccot}(\alpha) / \pi$ and $\mathcal{R}( \gamma_{0, \beta }) = \operatorname{Q}(\beta)$, respectively.

Lemma~\ref{lemma: classifier}, directly leads to the following theorem:
\begin{theorem}[DCP solution]
\label{thm:DCP-solution}
For $\operatorname{SNR}>0$ and a target error $P_e$ in the non-trivial Pareto range
\begin{equation}
\label{eq: pareto range Pe}
Q\!\left(\sqrt{\operatorname{SNR}}\right)\ \leq \ P_e\ \leq \ \operatorname{arccot}\!\left(\sqrt{\operatorname{SNR}}\right) / \pi,
\end{equation}
there exists a unique optimal solution $(\alpha^\star,\beta^\star)\in\mathbb{R}_+^2$ to \eqref{eq:DCP} with $\alpha^\star$ the solution of the non-linear equation
\begin{equation}
\label{eq: design}
    \Phi_\mathrm{SN}\left(  \beta^\star  / \sqrt{  1 + ({\alpha^\star})^2 }; \; \alpha^\star \right)  = P_e
\end{equation}
with $
\beta^\star = -\alpha^\star\sqrt{ 2 /\pi} + \sqrt{\operatorname{SNR}-({\alpha^\star})^2\!(1- 2 /\pi )}$.
\end{theorem}

Theorem~\ref{thm:DCP-solution} provides the design rule of the system, which boils down to the solution of \eqref{eq: design}.
It is worth noting that Theorem~\ref{thm:DCP-solution} does not rely on Lemma~\ref{lemma: decoder} and does not require MSE computation. Nonetheless, similarly to \eqref{eq: pareto range Pe}, one can set $\alpha=\sqrt{\operatorname{SNR}}$ in \eqref{eq: mmse linear} and  $\beta=\sqrt{\operatorname{SNR}}$ in \eqref{eq: mmse nonlinear} to assess the Pareto range in terms of MSE.

\section{Anomaly Detection}
\label{sec: anomaly detection}

In contrast to binary classification, AD assumes that only the distribution of the normal class is known a priori. This leads to specialization of a one-class classification framework.

\subsection{Preliminaries}

According to the fundamental concentration assumption of AD, the normal signal ({\rm ok}) lies within a bounded high-density region. Therefore, such a region is modeled as the smallest density level set that has a probability of at least $\theta $ under $f_X$ \cite{Ruff_2021ProcIEEE}. Given $X \sim \mathcal N(0, \sigma_X^2)$, the high-density region is the interval $[ -T, T ]$ such that $\theta = \operatorname{Pr}\{|X| \leq T \} = 1 - 2Q (T/\sigma_X)$, or, equivalently, the normal samples $x^{\ok}$ are distributed as a truncated Gaussian $
f_{X^{\ok}}(x) = \phi\left( x / \sigma_X \right) / (\sigma_X  \theta) $ for $ x \in [-T,T]$.
In normalized variables $\tilde{x} = x/\sigma_X$,  and $t = T/\sigma_X$, this becomes 
$f_{\tilde{X}^{\ok}}(\tilde{x}) = 
 \phi\left(\tilde{x}\right) / \theta $ for $ \tilde x \in [-t, t]$ .

With this model of the normal signal, a natural choice is to draw anomalies ({\rm ko}) from the Gaussian tails
$f^{\tails}_{\tilde{X}^{\ko}}(\tilde{x}) = 
\phi\left(\tilde{x}\right) / (1-\theta) $ for $|\tilde{x}|> t$.
To model the presence of rare, non-Gaussian anomalies, we also introduce an unknown contamination distribution supported in the tails, $f^{\unknown}_{\tilde{X}^{\ko}}(\tilde{x})$ with $\mathrm{supp}\,f^{\unknown}_{\tilde{X}^{\ko}}\subset\{|\tilde{x}|>t\}$. 
The overall signal $\tilde X$ is a mixture of the form
\begin{align}
f_{\tilde{X}}(\tilde{x}) &= \left(1 - \epsilon \right) \left[\theta f_{\tilde{X}^{\ok}}(\tilde{x}) + \left(1 - \theta \right) f^{\tails}_{\tilde{X}^{\ko}}(\tilde{x}) \right] + \epsilon  f^{\unknown}_{\tilde{X}^{\ko}}(\tilde{x})\nonumber\\
& = \pi^{\ok} f_{\tilde{X}^{\ok}}(\tilde{x}) + \pi^{\ko} f_{\tilde{X}^{\ko}}(\tilde{x}), 
\end{align}
with $\epsilon \ll 1$, $\pi^{\ok} = \theta (1 - \epsilon)$, $\pi^{\ko}=1-\pi^{\ok}$ and anomaly-class conditional PDF
\begin{equation}
    \label{eq: anomaly model}
    f_{\tilde{X}^{\ko}}(\tilde{x}) = \left( 1 - \tau \right)  f^{\tails}_{\tilde{X}^{\ko}}(\tilde{x}) + \tau f^{\unknown}_{\tilde{X}^{\ko}}(\tilde{x}),
\end{equation}
where $\tau = \epsilon / \pi^{\ko}$.
From a taxonomy perspective, \eqref{eq: anomaly model} captures both samples that lie in the low-probability region under the law that describes normality (outliers) and observations from a distinct distribution (anomalies).

The transmitter encodes the signal it observes as
\begin{equation}
\label{eq: AD encoding}
\tilde{g}(\tilde{x}) = 
\begin{cases}
\alpha \tilde{x}, & |\tilde{x}| \leq t \;\;\; (\ok), \\
\beta \tilde{x} + \delta \cdot \operatorname{sign}(\tilde{x}), & |\tilde{x}|>t \;\;\; (\ko),
\end{cases}
\end{equation}
where $\alpha = A\sigma_X/\sigma_Z$, $\beta = B/\sigma_Z$, and $\delta = D/\sigma_Z$. When $\beta=0$, anomalies are encoded to maximize discriminability from the normal signal. A nonzero $\beta$ allows to convey information regarding the anomaly beyond its sign.

With the encoding in \eqref{eq: AD encoding}, if we neglect the unknown distribution contribution, the normalized power constraint reads as 
\begin{equation}
    \label{eq: power AD}
    \mathbb{E} \left[ \tilde g^2(x) \right] = \theta \operatorname{SNR}^{\ok} + (1-\theta)\operatorname{SNR}^{\ko} \leq \operatorname{SNR},
\end{equation}
where 
\begin{align}
    \operatorname{SNR}^{\ok} & = 
    \alpha^2 \left[ 1 - 2 t  \phi \left(t \right) / \theta \right], \\
    \operatorname{SNR}^{\ko} & = 
    2 \left[ \beta(t \beta+2\delta)  \phi \left(t \right)
    + (\beta^2+\delta^2) Q(t) \right] / \left(  1-\theta \right).
\end{align}
Proofs of this and the results presented in the following subsection are provided in Appendix~B.

\begin{figure*}[t!]
\centering
  \subfloat[]{\includegraphics[]{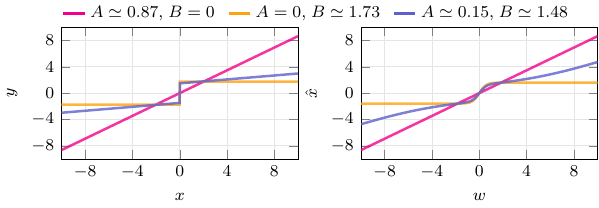} \label{fig: classification encodings}}
  \hspace{1cm}
  \subfloat[]{\includegraphics[]{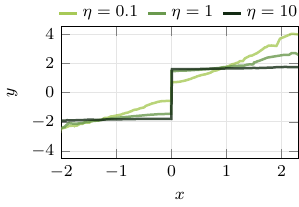} \label{fig: MLGD classification encodings}}
    \caption{(a) Examples of encoding (left) and decoding (right) mappings in the case of binary classification with $\sigma_X = 2$, $\sigma_Z = 1$, and power constraint $P = 3$, for three different feasible combinations of $A$ and $B$. 
  (b) Meta-learned encoding mappings in the case of binary classification with $\sigma_X = 1$, $\sigma_Z = 0.63$, and $P = 3$, for three different values of the classification loss weight $\eta$. The parameter $\eta$ is used as a Lagrange multiplier in the optimization objective to balance classification performance and reconstruction fidelity.}
  \label{fig:combined_classification_figures}
\end{figure*}

\captionsetup[subfigure]{width=0.97\columnwidth}
\begin{figure*}
\centering
  \subfloat[DCP Pareto fronts in case of binary classification for different values of $P$ for $\sigma_X = 2$ and  $\sigma_Z = 0.5$ (left),  $\sigma_Z = 1$ (right).]{\includegraphics[]{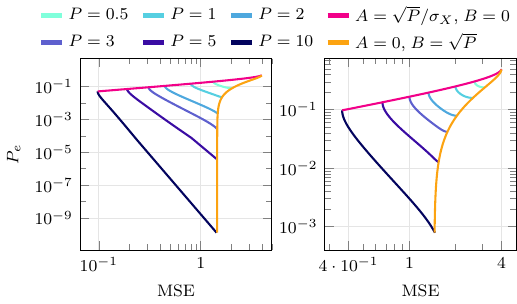} \label{fig: classification pareto}}
  \hfill
  \subfloat[DCP curves in case of AD for different values of $P$ for $\sigma_X = 2$, $T=2 \sigma_X$, $\sigma_Z = 0.5$ (left),  $\sigma_Z = 1$ (right) where $A_{\max} = \sqrt{P} / (\sigma_X\theta - 2 T  \phi(T/\sigma_X))$ and $D_{\max} = \sqrt{P} / (1 - \theta)$.]{\includegraphics[]{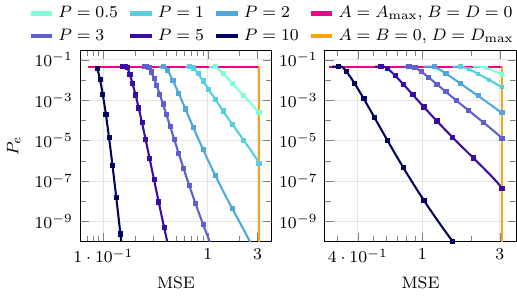} \label{fig: ad pareto}}
  \caption{Distortion-classification-power trade-off in case of binary classification (a) and anomaly detection (b).}
  \label{fig:combined_pareto_figures}
\end{figure*}

\subsection{Main results}
We first derive the optimal decoder for the normal signal.
\begin{lemma}[Decoder \ok]
    \label{lemma: decoder ad}
     For $x^\ok \sim f_{\tilde{X}^{\ok}}(\tilde{x})$ encoded with $g(x^\ok) = \alpha x^\ok$, the MMSE estimator decoder is
     \begin{align}
     \label{eq: decoder ad}
        &\tilde h_{\alpha}(\tilde w)=
        \frac{\alpha \tilde w}{\varsigma^2} - \frac{1}{\varsigma} \frac{\phi\left(\alpha \tilde w / \varsigma - t \varsigma  \right ) - \phi\left(\alpha \tilde w / \varsigma +  t \varsigma  \right )}{Q\left( \alpha \tilde w / \varsigma - t\varsigma \right) - Q\left( \alpha \tilde w / \varsigma + t\varsigma  \right)} ,
        \end{align}
    where $\tilde{w} = w/\sigma_Z$ and $\varsigma = \sqrt{1 + \alpha^2}$.
\end{lemma}
A notable special case of Lemma~\ref{lemma: decoder ad} is for $t \to +\infty$: the second term in \eqref{eq: decoder ad} vanishes and we obtain the MMSE decoder for the Gaussian signal. 

The next results concern a log-likelihood-based detector \cite{Marchioni_TSMC2024}, which computes the anomaly score $S(\tilde w) = -\log{ f_{\tilde W^\ok}(\tilde w) } $, and its performance in terms of the false positive rate $\operatorname{FPR} \left( \psi \right) = \Pr \{ \hat{C} = 1 \mid C = 0 \} $ and the false negative rate $\operatorname{FNR} \left( \psi \right) = \Pr \{ \hat{C} = 0 \mid C = 1 \} $, given a detection threshold $\psi$.
\begin{lemma}[Detector]
    \label{lemma: detector}
    The log-likelihood score is a monotonically increasing transformation of $|\tilde w|$, so that the detection reduces to
        $\hat{C} = \gamma_{\alpha, \beta, \delta}(\tilde{w}) = \mathbf{1}\{ |\tilde{w}| > \psi \}$.
\end{lemma}
Now, by relying on the bivariate standard normal cumulative distribution function with a correlation coefficient $\rho \in (-1, 1)$, which can be expressed as \cite[Identity 10,010.2]{Owen_1980} 
\begin{equation}
    \Phi_2(\xi, \omega; \rho) = \textstyle{\int^{\xi}_{-\infty} Q \left[ \left(\rho s - \omega \right) / \sqrt{1-\rho^2}  \right] \phi(s) ds},
\end{equation}
we can state the following lemma.
\begin{lemma}[Risk]
    \label{lemma: detector performance}
    The risk of the detector $\gamma_{\alpha, \beta, \delta}(\tilde{w})= 1\left\{ |\tilde{w}| > \psi \right\} $ is
     \begin{equation}
    \mathcal{R}(\gamma_{\alpha, \beta, \delta}) =  \pi^{\ok} \operatorname{FPR} \left( \psi \right) +  \pi^{\ko} \operatorname{FNR}\left( \psi \right),
    \end{equation}
    where 
    \begin{equation}
        \operatorname{FNR}\left( \psi \right) = \left( 1 - \tau \right)  \operatorname{FNR}^{\tails}\left( \psi \right) + \tau \operatorname{FNR}^{\unknown}\left( \psi \right) ,
    \end{equation}
    and we can explicate
    \begin{align}
        \operatorname{FPR} \left( \psi \right) &= 2 \Phi_2 \left(-\psi / \varsigma, t; -\alpha / \varsigma \right)/ \theta\\
        \notag
        & \quad \quad - 2 \Phi_2 \left( -\psi / \varsigma, -t; -\alpha / \varsigma \right) / \theta\\
        \label{eq: fnr tails}
        \operatorname{FNR}^{\tails} \left( \psi \right) &= 2 \Phi_2 \left(\kappa_+ / \vartheta, -t; -\beta / \vartheta \right) / (1-\theta)\\
        \notag
        & \quad \quad - 2 \Phi_2 \left( \kappa_- / \vartheta, -t ; -\beta / \vartheta \right) / (1-\theta),
    \end{align}
    with $\varsigma = \sqrt{1 + \alpha^2}$, $\vartheta = \sqrt{1 + \beta^2}$, $\kappa_\pm = \delta \pm \psi$. 

\end{lemma}

We can consider two extreme cases; \textit{i}) $\beta = 0$, the anomalies are sign encoded and $\operatorname{FNR} \left( \psi \right) = Q( \delta -  \psi )  -  Q ( \delta + \psi  )$. 
 \textit{ii})
$\epsilon = 0$, we have no contamination and \eqref{eq: fnr tails} reduces to  $\mathcal{R}(\gamma_{\alpha, \beta, \delta}) = \theta \operatorname{FPR} ( \psi ) + ( 1 - \theta ) \operatorname{FNR}^{\tails}( \psi )$. Moreover, if we assume the unknown anomaly is uniform over the interval $[-m, -t] \cup [t, m]$ with $m > t$, i.e., $ f^\unknown_{\tilde X^\ko}(\xi) =
1 / [2(m - t)] \, \mathbf{1}_{[-m, -t] \cup [t, m]}(\xi)$, for $\beta \neq 0$, we obtain
\begin{align}
 \operatorname{FNR}^\uniform \left( \psi \right) &= \left[G(\kappa_+ + \beta m) - G(\kappa_+ + \beta t) \right] / \nu\\
\notag
& \quad \quad+ \left[G(\kappa_- + \beta t) - G(\kappa_- + \beta m) \right] / \nu
\end{align}
where $G ( \xi ) =  \xi  Q (-\xi) + \phi( \xi) $ and $\nu = \beta (m-t)$.

Finally, to compute the risk, one still needs to set a threshold $\psi$. One solution is the Bayes optimal threshold $\psi^\star$ obtained by solving the stationarity condition $\partial / \partial \psi[\mathcal{R}(\gamma_{\alpha, \beta, \delta})] = 0$. With Lemma~\ref{lemma: detector performance} and \eqref{eq: power AD}, we derive the following proposition.

\begin{prop}[DCP solution]
\label{prop: design ad}
For $\epsilon=0$ and $\beta=0$ a heuristic solution ($\alpha^\star$, $\delta^\star$, $\psi^\star$) to \eqref{eq:DCP} satisfies the system
\begin{equation}
\begin{cases}
    2 \left[\Phi_2 \left(-\psi^\star / \varsigma, t; -\alpha^\star / \varsigma \right) - \Phi_2 \left( -\psi^\star / \varsigma, -t; -\alpha^\star / \varsigma \right) \right] \\
     + \left( 1 - \theta \right)\left[ Q( \delta -  \psi )  -  Q ( \delta + \psi  ) \right] = P_e  \\
     2\phi\left(\psi^\star / \varsigma\right) \left[ Q\left(\alpha^\star \psi^\star / \varsigma  + t\varsigma  \right) - Q \left(\alpha^\star \psi^\star / \varsigma  - t\varsigma \right) \right] / \varsigma\\
+ \left( 1 - \theta \right) \left[ \phi \left( \delta + \psi^\star \right)
+ \phi \left( \delta - \psi^\star \right) \right] = 0,
\end{cases}
\end{equation}
where $\delta = \sqrt{\{\operatorname{SNR} - (\alpha^\star)^2 [\theta - 2 t \phi( t ) ]\}/(1 - \theta)}$ and $\varsigma = \sqrt{1 + (\alpha^\star)^2}$.
\end{prop}

Since $\beta=0$, Proposition~\ref{prop: design ad} provides a design rule that preserves the fidelity of $\ok$ signal while maximizing detection, consistent with prior work on AD under compressed transmission \cite{Marchioni_TSMC2024, Enttsel_2024AICAS, Enttsel_EUSIPCO2024}. Setting $\epsilon=0$ reflects the common case where the anomaly distribution is unknown at design time. With $\beta=0$ the FNR depends only on $(\delta,\psi)$ and not on the anomaly law, so performance will not deteriorate in the presence of actual anomalies, provided that the prior $\epsilon$ is negligible.

\section{Numerical Results}
\label{sec:num results}

In this section, we present numerical results that validate the theoretical analysis and illustrate the trade-off between distortion, classification error, and transmission power.

For binary classification, Fig.~\ref{fig: classification encodings} provides a visual comparison of three representative piecewise linear encoders derived from Theorem~\ref{thm:DCP-solution} (left) and their corresponding decoders defined by Lemma~\ref{lemma: decoder} (right). The magenta curve shows the linear encoder–decoder optimized for reconstruction, the orange curve the sign–tanh pair optimized for classification, and the blue curve represents a linear combination of the previous mappings.

For comparison, Fig.~\ref{fig: MLGD classification encodings} shows encoding functions obtained using the modern learning-based optimizer MLGD \cite{Yang_WCSP2023, MLAM}. The meta-optimizer optimizes the encoder\footnote{Two-layer LSTM (20 hidden units), trained for 6500 epochs on batches of 300 samples, with learning rate $0.001$, $3\%$ warm-up, and cosine annealing to $2 \times 10^{-4}$.}, while the MMSE decoder is computed by closed form expressions provided in \cite{Akyol_2014TIT}[Theorem~3], where the signal PDF is evaluated numerically with kernel density estimation. We can observe that without restrictions on the function space, the numerically optimized encoder resembles a superposition of a step and a nonlinear function, closely mirroring piecewise-linear designs in Fig.~\ref{fig: classification encodings}.

\begin{table}[h!]
\centering
\begin{tabular}{c|ccccc}
\toprule
$P_e$ & 0.1 & 0.058 & 0.038 & 0.013 & 0.0071 \\
\midrule
MSE (piecewise-linear) & 0.12 & 0.13 & 0.15 & 0.22 & 0.28 \\
MSE (MLGD)              & 0.16 & 0.15 & 0.15 & 0.21 & 0.25 \\
\bottomrule
\end{tabular}
\vspace{0.5em}
\caption{Distortion values for different $P_e$ values, with $P = 3$, $\sigma_X=1$ and $\sigma_Z=0.63$ for the piecewise-linear and MLGD schemes.}
\label{tab:distortion_pe_2x6}
\end{table}

In Fig.~\ref{fig: classification pareto}, we illustrate the DCP Pareto fronts for two channel noise levels, as defined by Theorem~\ref{thm:DCP-solution}. 
The magenta curve corresponds to the region achievable with linear encoding, the orange curve to the region achievable with the sign-only scheme, and the blue curves, obtained with piecewise linear encoders, illustrate low-complexity intermediate trade-offs between distortion and classification error under fixed power.

In Table \ref{tab:distortion_pe_2x6}, we provide a comparison between the piecewise-linear and learned results. MLGD provides better performance towards the central non-trivial region of the tradeoff, but at the cost of computational complexity $O(mn^2)$, where $m,n$ are the batch sizes of the noise and the source samples. 

Similarly, for AD, in Fig.~\ref{fig: ad pareto}, we plot the DCP curves defined by Proposition~\ref{prop: design ad} for two levels of channel noise. 
In this setting, the magenta curve corresponds to allocating all available power to the normal signal, thereby optimizing reconstruction, while the orange curve represents the case where all power is allocated to the anomalous signal, maximizing detectability. The blue curves, obtained with piecewise linear encoders, interpolate between these extremes, characterizing the reconstruction–detection trade-off.

Although Proposition~\ref{prop: design ad} assumes no contamination ($\epsilon = 0$), we empirically verify that the design remains effective under mild contamination. The square markers in Figure~\ref{fig: ad pareto}, corresponding to $\epsilon = 10^{-3}$, follow the continuous curves, confirming the robustness of the proposed approach.

\section{Conclusion}
We investigated the fundamental trade-off between source distortion, classification accuracy, and transmission power in JSCC for Gaussian sources with disjoint classes. Our analysis is carried out within the function space of locally first-order encoder mappings. We examined both binary classification and anomaly detection setups, deriving closed-form solutions for the decoder and the classifier. For the piewise-linear encoder, we showed that the optimal mapping is obtained as the solution of a tractable non-linear algebraic equation. In addition, we presented numerical comparisons with a meta-learning-based optimization method, and visualized the achievable trade-offs through the Pareto range and the corresponding Pareto front. Our results show that, even though the optimality of the piecewise-linear scheme does not hold, unlike the corner cases of pure reconstruction or pure classification, it provides a low-complexity near-optimal alternative for zero-delay transmission. Future work will consider bandwidth mismatch scenarios and more general source and class distributions.

\bibliographystyle{IEEEbib}
\bibliography{strings,refs}

\vfill\pagebreak
\section*{Appendix A}

\begin{proof}[Proof of Lemma~\ref{lemma: decoder}]
\label{proof: of lemma decoder}
From \cite{Akyol_2014TIT}[Theorem~3], given an optimal encoder $\tilde g$, the MMSE decoder can be expressed as
\begin{align}
      \tilde h_{\alpha,\beta}(\tilde w)&= \frac{\int_{-\infty}^{+\infty} t f_{\tilde Z}(\tilde{w} - \tilde g(t)) f_{\tilde X} (t)\, dt}{\int_{-\infty}^{+\infty} f_{\tilde Z}(\tilde w- \tilde g(t)) f_{\tilde X}(t) \, dt} \\
       \label{eq: dec 1}
       &= \frac{\int_{-\infty}^{+\infty} t \phi\left(\tilde g(t) - \tilde w) \right) \phi\left(t\right) dt}{\int_{-\infty}^{+\infty}  \phi\left( \tilde g(t) - \tilde w\right) \phi\left( t \right) dt}
\end{align}
where the last equality is due to $f_{\tilde X}(\xi) = f_{\tilde Z}(\xi) = \phi(\xi) = \phi(-\xi)$.
Now, since $\tilde g(t) = \alpha t + \beta \operatorname{sign}( t)$, we define the anti-derivatives
\begin{align}
    N_\pm & = \int t \phi\left( \alpha t \pm \beta - \tilde w \right) \phi\left(t\right) dt \\
    D_\pm & = \int \phi\left( \alpha t \pm \beta - \tilde w \right) \phi\left(t\right) dt
\end{align}
to express the decoding functions as
\begin{align}
        \label{eq: dec 2}
       \tilde h(\tilde w)_{\alpha,\beta}&= \frac{\left[N_- \right]_{-\infty}^0 + \left[ N_+ \right]_0^\infty}{\left[D_- \right]_{-\infty}^0 + \left[D_+ \right]_{0}^\infty}.
\end{align}
If we define $\varsigma = \sqrt{{1 + \alpha^2}}$, $\nu_\pm = \left( \pm \beta - \tilde w \right)$ and use \cite[Identities 111; 110]{Owen_1980}\footnote{Identity 110 wrongly omits the $\alpha$ term highlighted in red.}, the anti-derivatives can be written in terms of $\Phi(x) = Q(-x)$
\begin{align}
\label{eq: primitive N}
N_\pm &= \int t \phi\left( \alpha t +  \nu_\pm  \right) \phi\left(t\right) dt \\
&= \frac{-1}{\varsigma^2}\phi\left( \frac{\nu_\pm}{\varsigma}\right) \left[\phi\left(t \varsigma + \frac{\alpha \nu_\pm }{\varsigma} \right )
+ \frac{\alpha \nu_\pm }{\varsigma}\Phi\left( t \varsigma + \frac{\alpha  \nu_\pm }{\varsigma} \right) \right]\\
\label{eq: primitive D}
D_\pm &=\int \phi\left( \alpha t +  \nu_\pm  \right) \phi\left(t\right) dt = \frac{1}{\varsigma}\phi\left( \frac{\nu_\pm}{\varsigma}\right)\Phi\left(t \varsigma + \frac{\textcolor{red}{\alpha}\nu_\pm }{\varsigma} \right).
\end{align}
Since $\phi(\pm \infty) = 0$, $\Phi(-\infty) = 0$ and $\Phi(\infty) = 1$ the numerator becomes
\begin{align}
        &\left[N_- \right]_{-\infty}^0 + \left[N_+ \right]_{0}^\infty \\
        &= -\frac{1}{\varsigma^2}\phi\left( \frac{\nu_-}{\varsigma}\right) \left[\phi\left(\frac{\alpha \nu_- }{\varsigma} \right ) + \frac{\alpha \nu_- }{\varsigma}\Phi\left( \frac{\alpha  \nu_- }{\varsigma} \right)\right]\\
        &- \frac{1}{\varsigma^2}\phi\left( \frac{\nu_+}{\varsigma}\right) \left[\frac{\alpha \nu_- }{\varsigma}-\phi\left(\frac{\alpha \nu_+ }{\varsigma} \right ) - \frac{\alpha \nu_+ }{\varsigma}\Phi\left( \frac{\alpha  \nu_+ }{\varsigma} \right)\right] \\
         &=-\frac{1}{\varsigma^2}\left[\phi\left( \frac{\nu_+}{\varsigma}\right) \phi\left( \frac{ \alpha \nu_+}{\varsigma}\right)  + \frac{\alpha \nu_- }{2\varsigma} \phi_\mathrm{SN}\left( \frac{\nu_-}{\varsigma}; \alpha \right) \right]\\
        &- \frac{1}{\varsigma^2}\left[\phi\left( \frac{\nu_-}{\varsigma}\right) \phi\left( \frac{ \alpha \nu_-}{\varsigma}\right) + \frac{\alpha \nu_+ }{2\varsigma} \phi_\mathrm{SN}\left( \frac{\nu_+}{\varsigma}; -\alpha \right) \right]
\end{align}
where in the last equation we exploit the facts that $\Phi(\xi) = Q(-\xi)$, $1 - \Phi(\xi) = Q(\xi)$ paired with the definition in \eqref{eq: skew pdf}. Now it is straightforward to show that $\phi\left( \nu_\pm / \varsigma \right)\phi\left(  \alpha \nu_\pm / \varsigma \right)  = \exp\left( -\nu^2_\pm / 2\right) / ( 2 \pi) $, hence ultimately
\begin{align}
        &\left[N_- \right]_{-\infty}^0 + \left[N_+ \right]_{0}^\infty \\
        &=-\frac{1}{\varsigma^2}\left[\frac{1}{\sqrt{2 \pi}}\phi\left( \nu_-\right) +
        \frac{\alpha \nu_- }{2\varsigma} \phi_\mathrm{SN}\left( \frac{\nu_-}{\varsigma}; \alpha \right) \right]\\
        &- \frac{1}{\varsigma^2}\left[-\frac{1}{\sqrt{2 \pi}}\phi\left( \nu_+ \right) + \frac{\alpha \nu_+ }{2\varsigma} \phi_\mathrm{SN}\left( \frac{\nu_+}{\varsigma}; -\alpha \right) \right].
\end{align}
 Similarly,
\begin{align}
       &\left[D_- \right]_{-\infty}^0 + \left[D_+ \right]_{0}^\infty = \frac{1}{\varsigma}\phi\left( \frac{\nu_-}{\varsigma}\right) \Phi\left(\frac{\alpha \nu_- }{\varsigma} \right) \\
       &+ \frac{1}{\varsigma}\phi\left( \frac{\nu_+}{\varsigma}\right) \left[ 1 - \Phi\left(\frac{\alpha \nu_+ }{\varsigma} \right) \right]\\
       &=\frac{1}{2\varsigma} \left[\phi_\mathrm{SN}\left( \frac{\nu_-}{\varsigma}; \alpha  \right) +\phi_\mathrm{SN}\left( \frac{\nu_+}{\varsigma}; - \alpha  \right)\right].
\end{align}
With this, \eqref{eq: dec 2} can be written as
\begin{align}
       \tilde h(\tilde w)_{\alpha,\beta}&= \frac{2}{\sqrt{2 \pi}\varsigma}\frac{\phi\left( \nu_+ \right) - \phi\left( \nu_- \right) }{\phi_\mathrm{SN}\left( \nu_- / \varsigma; \alpha \right) +\phi_\mathrm{SN}\left( \nu_+ / \varsigma; -\alpha \right)} \\
       & - \frac{\alpha}{\varsigma^2} \frac{\nu_- \phi_\mathrm{SN}\left( \nu_- / \varsigma; \alpha \right)  + \nu_+ \phi_\mathrm{SN}\left( \nu_+ / \varsigma; -\alpha \right) }{\phi_\mathrm{SN}\left( \nu_- / \varsigma; \alpha \right) +\phi_\mathrm{SN}\left( \nu_+ / \varsigma; -\alpha \right)}
\end{align}
Since
\begin{align}
\label{eq: sinh1}
&\phi\left( \nu_+ \right) - \phi\left( \nu_- \right)\\
&= \frac{1}{\sqrt{2 \pi}} \exp\left(- \frac{\tilde w^2 + \beta^2}{2}\right)\cdot \left[ \exp\left(\tilde w \beta \right) -  \exp\left( -\tilde w \beta \right)\right] \\
\label{eq: sinh2}
&= \frac{2}{\sqrt{2 \pi}} \exp\left(- \frac{\tilde w^2 + \beta^2}{2}\right) \sinh\left( \tilde w \beta \right)
\end{align}
and if we exploit the fact that $\phi_\mathrm{SN}(-x;\lambda) = \phi_\mathrm{SN}(x;-\lambda)$ to define $S_{\pm}(\tilde w)=\phi_\mathrm{SN}\left( \nu_\pm / \varsigma; \mp \alpha  \right) = \phi_{\mathrm{SN}}\left[\left(\tilde w\mp\beta \right)/ \varsigma;\pm\alpha\right]$
\begin{align}
   \tilde h(\tilde w)_{\alpha,\beta}&= \frac{2}{ \pi \varsigma}\frac{\exp\left(- \frac{w^2 + \beta^2}{2}\right) \sinh\left(w \beta \right)}{S_-\left(\tilde w \right)+S_+\left(\tilde w \right)} \\
   & + \frac{\alpha}{\varsigma^2} \frac{\left(\tilde w + \beta \right) S_-\left(\tilde w \right) + \left(\tilde w - \beta \right) S_+\left(\tilde w \right) }{S_-\left(\tilde w \right)+S_+\left(\tilde w \right)}
\end{align}
After replacing $\varsigma = \sqrt{{1 + \alpha^2}}$, we get get exactly the expression in \eqref{lemma: decoder}.

\end{proof}

\begin{proof}[Proof of the corner cases for Lemma~\ref{lemma: decoder}]
For $\beta = 0$, $\sinh(0) = 0$, and exploiting the definition of $S_{\pm}(\tilde w)$
\begin{align}
    \tilde h(\tilde w)_{\alpha, 0}&= \frac{\alpha}{1 +\alpha^2} \frac{\tilde w \phi_{\mathrm{SN}}\left[ \frac{\tilde w}{\sqrt{1 + \alpha^2}} ; - \alpha\right]+  \tilde w  \phi_{\mathrm{SN}}\left[ \frac{\tilde w}{\sqrt{1 + \alpha^2}} ; \alpha\right] }{\phi_{\mathrm{SN}}\left[ \frac{\tilde w}{\sqrt{1 + \alpha^2}} ; - \alpha\right] + \phi_{\mathrm{SN}}\left[ \frac{\tilde w}{\sqrt{1 + \alpha^2}} ; \alpha\right]} \\
    &= \frac{\alpha}{1 +\alpha^2} \tilde w
\end{align}
For $\alpha = 0$, since $\phi_\mathrm{SN}(\xi;0) = 2\phi(\xi)$, $S_{\pm}(\tilde w) = \phi(\tilde w\mp\beta)$ and 
\begin{align}
    \tilde h(\tilde w)_{0, \beta}&= \frac{2}{ \pi}\frac{\exp\left(- \frac{w^2 + \beta^2}{2}\right) \sinh\left(w \beta \right)}{\phi(\tilde w +\beta)+\phi(\tilde w - \beta)}.
\end{align}
With
\begin{align}
&\phi(\tilde w +\beta)+\phi(\tilde w - \beta)\\
&= \frac{1}{\sqrt{2 \pi}} \exp\left(- \frac{\tilde w^2 + \beta^2}{2}\right)
 \left[ \exp\left(\tilde w \beta \right) +  \exp\left( -\tilde w \beta \right)\right] 
\end{align}
and exploiting \eqref{eq: sinh2} together with the definition of $\tanh$
\begin{align}
    \tilde h(\tilde w)_{0, \beta}&= \sqrt{\frac{2}{ \pi}}\frac{\exp\left(\tilde w \beta \right) -  \exp\left( -\tilde w \beta \right)}{\exp\left(\tilde w \beta \right) +  \exp\left( -\tilde w \beta \right)} = \sqrt{\frac{2}{ \pi}} \tanh\left(\tilde w \beta\right).
\end{align}
\end{proof}

\begin{proof}[Expression of the non-normalized decoder]
From \cite{Akyol_2014TIT}[Theorem~3], given an optimal encoder $ g$, the MMSE decoder can be expressed as
\begin{align}
    h_{A, B}( w)&= \frac{\int_{-\infty}^{+\infty} \xi  f_{ Z}\left[w -  g(t) \right] f_{ X} ( \xi )\, d\xi }{\int_{-\infty}^{+\infty} f_{ Z}\left[ w-  g(\xi ) \right] f_{ X}(\xi ) \, d\xi }  
\end{align}
Since for a Gaussian variable $V \sim f_V (v) = \phi\left( v / \sigma_V \right)/\sigma_V$
\begin{align}
    h_{A, B}( w)&= \frac{\int_{-\infty}^{+\infty} \xi  \phi\left[  w / \sigma_Z -  g(\xi ) / \sigma_Z \right] \phi\left( \xi  / \sigma_X \right)\, d\xi }{\int_{-\infty}^{+\infty} \phi\left[ w / \sigma_Z - g(\xi ) / \sigma_Z \right] \phi\left( \xi  / \sigma_X \right) \, d\xi }  
\end{align}
If we replace $\xi = \sigma_X t$
\begin{align}
    h_{A, B}( w)&= \frac{\int_{-\infty}^{+\infty} \sigma_X t \phi\left[  w / \sigma_Z - g( \sigma_X t ) / \sigma_Z \right]\phi\left( t \right)\, d t  }{\int_{-\infty}^{+\infty} \phi\left[  w / \sigma_Z - g( \sigma_X t ) / \sigma_Z \right] \phi\left( t\right) \, d t} \\
    & = \sigma_X \frac{\int_{-\infty}^{+\infty} t \phi\left[  w / \sigma_Z - \tilde g(  t) \right]\phi\left( t \right)\, d t  }{\int_{-\infty}^{+\infty} \phi\left[ w / \sigma_Z - \tilde g(  t) \right] \phi\left( t\right) \, d t} 
\end{align}
where $ g( \sigma_X t ) / \sigma_Z = A \sigma_X t / \sigma_Z + B \operatorname{sign}(t) / \sigma_Z = \tilde g(  t)$. Now from \eqref{eq: dec 1} it follows that $h_{A, B}( w) = \sigma_X \tilde h_{\alpha,\beta} (w / \sigma_Z ) $.

\end{proof}

\begin{proof}[Proof of Lemma~\ref{lemma: classifier}]
Due to symmetry the optimal Bayesian binary classifier is $\gamma_{\alpha,\beta}(\tilde w) = 1\left\{ \tilde w > 0 \right\}$ characterized by the risk
\begin{align}
    &\mathcal{R}(\gamma_{\alpha, \beta} ) = \Pr\left\{\gamma \left(\tilde W \right) \neq C \right\} 
    = \Pr\left\{\tilde W > 0 \wedge \tilde X < 0 \right\} \\
    &+ \Pr\left\{\tilde  W < 0 \wedge \tilde X > 0 \right\}
    = 2\Pr\left\{\tilde W > 0 \wedge \tilde X < 0 \right\}
\end{align}
with the last equality due to symmetry. Now if we exploit the fact that $f_{\tilde X}(t) = f_{\tilde Z}(t) = \phi(t)$ and the definition of $\tilde g$
\begin{align}
    \mathcal{R}(\gamma_{\alpha, \beta} ) &= 2\int_{-\infty}^0 \Pr\left\{ \tilde W > 0 \mid \tilde X = t \right\} f_{\tilde X}(t) \, dt\\
    &= 2\int_{-\infty}^0 \Pr\left\{\tilde Z > - \tilde g(t)\right\} \phi(t) \, dt\\
    &= 2\int_{-\infty}^0 \int_{- \tilde g(\xi)}^{\infty} \phi(\xi)\,d\xi \phi(t) \, dt \\ 
    &= 2\int_{-\infty}^0 Q\left(- \tilde g(t) \right) \phi(t) \, dt \\
    &= 2\int_{-\infty}^0 \Phi \left( \alpha t - \beta \right) \phi(t) \, dt
\end{align}
To solve the integral we rely on \cite[Identity 10,010.5]{Owen_1980}, so that
\begin{align}
\label{eq: risk 2}
    \mathcal{R}(\gamma_{\alpha, \beta}) &= \Phi \left(-\frac{\beta}{\sqrt{1 + \alpha^2}} \right) - 2 T \left( -\frac{\beta}{\sqrt{1 + \alpha^2}}, \;\alpha \right) \\
    &= \Phi_\mathrm{SN}\left( - \frac{\beta}{ \sqrt{\alpha^2 + 1}};  \; \alpha \right)    
\end{align}
where, in the last equality to get the definition of $\Phi_\mathrm{SN}(\xi)$, we exploit the fact that $Q(\xi) = \Phi(-\xi)$.
\end{proof}

\begin{proof}[Proof of the corner cases of Lemma~\ref{lemma: classifier}]
For $\beta=0$, we can use the identity $T\left(0, a\right)=\arctan\left(a\right) / (2\pi)$ \cite{Owen_1956AMS} so that \eqref{eq: risk 2} leads to
\begin{align}
    \mathcal{R}(\gamma_{\alpha, 0} ) &= \Phi\left(  0 \right) - 2 T\left(0,  \; \alpha \right) = \frac{1}{2} - \frac{1}{\pi} \arctan\left( \alpha \right)\\
    &= \frac{1}{2} - \frac{1}{\pi} \left[ \frac{\pi}{2} - \arctan\left( \frac{1}{\alpha} \right)\right ]\\
    \label{eq: risk beta=0}
    &= \frac{1}{\pi} \arctan\left( \frac{1}{\alpha} \right) = \frac{1}{\pi} \operatorname{arccot} \left( \alpha \right). 
\end{align}

In case $\alpha=0$, using the identity $\text{T}\left(h, 0\right) = 0$ \cite{Owen_1956AMS}, from \eqref{eq: risk 2} we get
\begin{align}
    \label{eq: risk alpha=0}
    \mathcal{R}(\gamma_{0, \beta} ) &= \Phi\left(   -\beta \right) = Q\left(  \beta \right).
\end{align}
\end{proof}

\begin{proof}[Proof of Theorem~\ref{thm:DCP-solution}]
Given a constraint in power $\textbf{E}_X [ g^2(x)] \leq P$,
for $ \tilde g( \tilde x) = \alpha \tilde x + \beta \text{sign} (\tilde x)$ it translates to:
\begin{align}
    \mathbb{E} \left[ \tilde g^2(x) \right] &= \alpha^2 \mathbb{E} \left[  \tilde x^2 \right] + 2 \alpha \beta  \mathbb{E} \left[  \tilde x \text{sign}( \tilde x) \right] \\
    &+ \beta^2 \mathbb{E} \left[ \text{sign}^2(\tilde x) \right] = \alpha^2  + 2 \alpha \beta \mathbb{E}\left[ |\tilde x| \right] + \beta^2 \\
    &= \alpha^2  + 2 \alpha \beta \sqrt{{\frac{2}{\pi}}}  + \beta^2 \leq \frac{P}{\sigma_Z^2} = \operatorname{SNR}
\end{align}
Now, pairing the power constraint with Lemma~\ref{lemma: classifier}, the DCP optimization problem can be reduced to
\begin{align}
    \label{eq: MSE_objective}
    (\alpha^\star, \beta^\star) &= \arg\min_{(\alpha, \beta) \in \mathbb{R}_+^2} \quad \mathrm{MSE}\left(x, \hat{x} \right) \\
    \text{s.t.} 
    \label{eq: Pe_constraint}
    \quad &\Phi_\mathrm{SN}\left(- \frac{\beta}{ \sqrt{\alpha^2 + 1}};  \; \alpha \right) \leq P_e, \\
    \label{eq: P_constraint}
    & \alpha^2 + 2\alpha \beta \sqrt{{\frac{2}{\pi}}} + \beta^2  \leq \operatorname{SNR}
\end{align}
Since the risk is minimized for $\alpha = 0$ and $\beta =\sqrt{\operatorname{SNR}}$, we have a constraint on the minimum achievable $P_e$ given $P$. In contrast, $P_e$ should not exceed the one achieved for $\beta=0$, otherwise the optimization is trivial, i.e. $
\alpha=\sqrt{\operatorname{SNR}}$ and $\beta = 0$. Hence, according to \eqref{eq: risk beta=0} and  \eqref{eq: risk alpha=0}, the Pareto frontier is defined for $P_e$ such that
\begin{equation}
\text{Q}\left( \sqrt{\operatorname{SNR}} \right)
      \leq P_e \leq
      \frac{1}{\pi} \operatorname{arccot} \left( \sqrt{\operatorname{SNR}} \right)
\end{equation}

Now, the constraint on $P_e$ in \eqref{eq: Pe_constraint} is active at the solution; otherwise, the budget on $P$ could be used to improve MSE by reducing $\beta$ and increasing $\alpha$.
The constraint on $P$ is also active; otherwise, it could be used to reduce MSE by increasing $\alpha$ and adjusting $\beta$ along the error-constraint curve. Since \eqref{eq: Pe_constraint} and \eqref{eq: P_constraint} can be considered with equality, DCP is a non-linear system of two equations in two variables.

Considering \eqref{eq: P_constraint} with equality, given $\alpha \leq \sqrt{\operatorname{SNR}}$, 
\begin{equation}
    \label{eq: beta}
    \beta = -\alpha\sqrt{\frac{2}{\pi}} + \sqrt{\operatorname{SNR}-\alpha^2\!\left(1-\frac{2}{\pi}\right)}
\end{equation}
that if plugged into \eqref{eq: Pe_constraint} with equality leads to a non-linear equation in $\alpha$
\begin{equation}
    \label{eq: non-linear equation}
    \Phi_\mathrm{SN}\left(\frac{\alpha \sqrt{2} - \sqrt{\pi \operatorname{SNR} - \alpha^2 \left( \pi - 2 \right)}}{\sqrt{ \pi \left( 1 + \alpha^2 \right)}}; \; \alpha \right)  = P_e
\end{equation}

Provided that the solution to this non-linear equation is unique, it gets us $\alpha^\star$, which when plugged in \eqref{eq: beta} gives us $\beta^\star$. $A^\star$ and $B^\star$ are then obtained from the definitions of $\alpha^\star$ and $\beta^\star$.

To prove the uniqueness, we show that the left-hand side of \eqref{eq: non-linear equation} is a strictly monotonic function, since its derivative is positive, for $\alpha > 0$. 

We start by defining
\begin{equation}
    u(\alpha)= \frac{\alpha \sqrt{2} - \sqrt{\pi \operatorname{SNR} - \alpha^2 \left( \pi - 2 \right)}}{\sqrt{ \pi \left( 1 + \alpha^2 \right)}}
\end{equation}
From the definitions of $\Phi_\mathrm{SN}(x, \lambda)$ and $T(h,a)$ \cite{Owen_1980} it follows
\begin{align}
    &\Phi_\mathrm{SN}\left[ u\left(\alpha\right); \alpha \right] = Q\left[-u\left(\alpha\right)\right] - 2\,T\left[u\left(\alpha\right),\alpha \right]  \\
    &=\Phi\left[u\left(\alpha\right)\right] - 2\int_0^\alpha \frac{\phi\left[u\left(\alpha\right)\right] \phi\left[ \xi u\left(\alpha\right) \right]}{1 + \xi^2} d\xi \\
    &= \Phi\left[u\left(\alpha\right)\right] - 2\int_0^\alpha J\left(\alpha, \xi \right) d\xi.
\end{align}
The definition of the functional $J$ allows us to compute the derivative w.r.t. $\alpha$ using the Leibniz integral rule \cite{Protter_1985Springer}
\begin{align}
    &\Phi'_\mathrm{SN}\left[u\left(\alpha\right); \alpha \right] = u'\left(\alpha\right)\phi\left[u\left(\alpha\right)\right] - 2J\left(\alpha, \alpha \right)\\
    &- 2\int_0^\alpha  J'\left(\alpha, \xi \right) d\xi.
\end{align}
Exploiting the derivative of the product, 
\begin{equation}
    J'\left(\alpha, \xi \right) = - u'\left(\alpha\right) u\left(\alpha\right)\phi\left[u\left(\alpha\right)\right] \phi\left[ \xi u\left(\alpha\right) \right] 
\end{equation}
so that
\begin{align}
&\Phi'_\mathrm{SN}\left[u\left(\alpha\right); \alpha \right] =  
u'\left(\alpha\right)\phi\left[u\left(\alpha\right)\right] - 2\frac{\phi\left[u\left(\alpha\right)\right] \phi\left[ \alpha u\left(\alpha\right) \right]}{1 + \alpha^2}\\
&+ 2u'\left(\alpha\right) u\left(\alpha\right)\phi\left[u\left(\alpha\right)\right] \int_0^\alpha \phi\left[ \xi u\left(\alpha\right) \right] d\xi.
\end{align}
Now with $t=\xi u\left(\alpha\right)$
\begin{equation}
    \int_0^\alpha \phi\left[ \xi u\left(\alpha\right) \right] d\xi = \frac{1}{u\left(\alpha\right)}\int_0^{\alpha u\left(\alpha\right)} \phi\left( t \right) dt = \frac{\Phi\left[ \alpha u\left(\alpha\right) \right] - \frac{1}{2}}{u\left(\alpha\right)}
\end{equation}
After some rearrangement, we can finally write
\begin{equation}
\label{eq: derivative lhs}
\Phi'_\mathrm{SN}\left[u\left(\alpha\right); \alpha \right] = 2 \phi\left[u\left(\alpha\right)\right] \cdot I\left( \alpha \right)
\end{equation}
with
\begin{equation}
    I\left( \alpha \right) = u'\left(\alpha\right) \Phi\left[ \alpha u\left(\alpha\right) \right] - \frac{ \phi\left[ \alpha u\left(\alpha\right) \right]}{1 + \alpha^2} 
\end{equation}
where
\begin{equation}
    u'(\alpha) = \frac{1}{\left( 1 + \alpha^2 \right)^{\frac{3}{2}}}\left[ \sqrt{\frac{2}{\pi}} + \frac{\alpha \left( \pi \operatorname{SNR} + \pi - 2\right)}{\sqrt{\pi^2 \operatorname{SNR} - \alpha^2 \pi \left(\pi - 2\right)}} \right] 
\end{equation}
Now that we have the derivative of $\Phi'_\mathrm{SN}[u(\alpha)]$, assessing its sign from \eqref{eq: derivative lhs} is still non-trivial. For this purpose, we construct a lower bound that we show to be positive.

Writing $\pi \operatorname{SNR}+ \pi - 2 = \pi \operatorname{SNR}+ \alpha^2( \pi - 2) + (\pi - 2 )(1 + \alpha^2) = v(\alpha) + (\pi - 2 )(1 + \alpha^2)$ leads to this equivalent form:
\begin{align}
    u'(\alpha) &= \sqrt{\frac{2}{\pi}} \frac{1}{\left( 1 + \alpha^2 \right)^{\frac{3}{2}}} \\
    &+ \frac{\alpha }{\sqrt{\pi} \left( 1 + \alpha^2 \right)^{\frac{3}{2}}}\left[  \sqrt{v(\alpha)}  + \frac{(\pi - 2 )(1 + \alpha^2)}{ \sqrt{v(\alpha)}}\right].  
\end{align}
Since the two terms in the squared brackets are positive, we can apply AM–GM inequality to them, which leads to
\begin{align}
    u'(\alpha) &\geq \sqrt{\frac{2}{\pi}} \frac{1}{\left( 1 + \alpha^2 \right)^{\frac{3}{2}}}+ \frac{\alpha 2 \sqrt{(\pi - 2 )(1 + \alpha^2)}}{\sqrt{\pi} \left( 1 + \alpha^2 \right)^{\frac{3}{2}}}\\
     &= \frac{\sqrt{2} + \alpha 2 \sqrt{(\pi - 2 )(1 + \alpha^2)}}{\sqrt{\pi} \left( 1 + \alpha^2 \right)^{\frac{3}{2}}}.
\end{align}
Having $\alpha \geq 0$, it is possible to show that 
\begin{equation}
    \sqrt{2} + \alpha 2 \sqrt{(\pi - 2 )(1 + \alpha^2)} \geq \sqrt{2 (1 + \alpha^2)}
\end{equation}
and therefore
\begin{align}
    u'(\alpha) &\geq \sqrt{\frac{2}{\pi}} \frac{1}{1 + \alpha^2}.
\end{align}
With this, we have the following lower bound of $I(\alpha)$
\begin{align}
\sqrt{\frac{2}{\pi}} \frac{1}{1 + \alpha^2} \Phi\left[ \alpha u\left(\alpha\right) \right] - \frac{ \phi\left[ \alpha u\left(\alpha\right) \right]}{1 + \alpha^2}  \leq I\left( \alpha \right). 
\end{align}
Now, since $u(\alpha) < 0\; \forall \; \alpha  < \sqrt{\operatorname{SNR}} $, we can rely on the bound
$\Phi(-x) = Q(x) \leq \sqrt{2 \pi } \phi(x)$ valid for $x > 0$ to obtain
\begin{align}
2\frac{ \phi\left[ \alpha u\left(\alpha\right) \right]}{1 + \alpha^2} - \frac{ \phi\left[ \alpha u\left(\alpha\right) \right]}{1 + \alpha^2} = \frac{ \phi\left[ \alpha u\left(\alpha\right) \right]}{1 + \alpha^2} \leq I\left( \alpha \right) 
\end{align}
$\forall \; |\alpha u\left(\alpha\right)| > 0$. This allows us to bound the derivative as
\begin{equation}
2\frac{ \phi^2\left[ \alpha u\left(\alpha\right) \right]}{1 + \alpha^2} \leq \Phi'_\mathrm{SN}\left[u\left(\alpha\right); \alpha \right] .
\end{equation}
Since $\phi(x) > 0 \; \forall \; |x| < \infty$, $\Phi'_\mathrm{SN}\left[u\left(\alpha\right); \alpha \right] > 0 \;\forall \; 0<|\alpha u\left(\alpha\right)|<\infty$ the left-hand side of \eqref{eq: non-linear equation} is a strictly monotonic function and the solution to \eqref{eq: non-linear equation} is unique.
\end{proof}

\begin{proof}[Proof of MMSE extreme cases.]
In general, the MMSE can be expressed as
\begin{align}
    &\mathrm{MMSE}\left( x, \hat {x}\right) = \mathbb{E} \left[ \left( x - \hat{x} \right)^2\right] = \mathbb{E} \left[ \left( x - h(w) \right)^2\right]\\
    &= \sigma_X^2 \mathbb{E} \left[ \left( \tilde x - \tilde h (\tilde w ) \right)^2\right] = \sigma_X^2 \mathrm{MMSE}\left( \tilde x, \breve {x}\right)
\end{align}
where with $\breve {x} = \tilde h (\tilde w )$ we define the reconstruction of $\tilde x$.
The two useful MMSE expressions involving scaled quantities are
\begin{align}
    \label{eq: MSE expectation}
&\mathrm{MMSE}\left( \tilde x, \breve {x}\right)
  = 1 - 2\mathbb{E} \left[ \tilde x \tilde h(\tilde w) \right] + \mathbb{E} \left[ \tilde h^2(\tilde w) \right]\\
    \label{eq: MSE integral1}
     &= 1  + \int_{-\infty}^{+\infty}  \int_{-\infty}^{+\infty}  \tilde h^2\left( \tilde w \right) \phi \left( \nu \right) \phi \left( t \right) d\nu dt \\
     \label{eq: MSE integral2}
     &- 2 \int_{-\infty}^{+\infty}  \int_{-\infty}^{+\infty} t \tilde h\left(  \tilde w\right)\phi \left( \nu \right) \phi \left( t \right) d\nu dt 
\end{align}
In case $\beta=0$, the equality \eqref{eq: MSE expectation} reduces to
\begin{align}
    &\mathrm{MMSE}\left( \tilde x, \breve {x}\right)
    =1 - \frac{2\alpha}{1 +\alpha^2} \mathbb{E} \left[ \tilde x\tilde w \right]+ \frac{\alpha^2}{\left(1 +\alpha^2\right)^2} \mathbb{E} \left[ \tilde w^2 \right] 
\end{align}
Since, $\mathbb{E}[\tilde x^2] = \mathbb{E}[\tilde z^2]=1$ and $\mathbb{E}[\tilde x\tilde z]=0$, we have $\mathbb{E}[ \tilde x\tilde w ] = \mathbb{E}[\tilde x ( \alpha \tilde x + \tilde z)] = \alpha$ and $\mathbb{E}[ \tilde w^2 ] =\mathbb{E}[ (\alpha \tilde x + \tilde z)^2 ]= \alpha^2 \mathbb{E}[ \tilde x^2 ] - 2 \alpha \mathbb{E}[\tilde x\tilde z] + \mathbb{E}[\tilde z^2] = \alpha^2 + 1$ so that
\begin{align}
    \mathrm{MMSE}\left( \tilde x, \breve {x}\right)& = 1 -  \frac{2 \alpha^2}{1 +\alpha^2}  + \frac{\alpha^2}{1 +\alpha^2}  = \frac{1}{1 +\alpha^2}
\end{align}
In case $\alpha=0$, exploiting symmetry, the expression in \eqref{eq: MSE integral1},~\eqref{eq: MSE integral2} yields
\begin{align}
    &\mathrm{MMSE}\left( \tilde x, \breve {x} \right)\\
    &= 1 + 2 \int_{0}^{+\infty} \int_{-\infty}^{+\infty}  \tilde h^2\left( \beta + \nu \right) \phi \left( \nu \right) \phi \left( t \right)d\nu dt\\
     & -4 \int_{0}^{+\infty} \int_{-\infty}^{+\infty}  t \tilde  h\left( \beta + \nu \right)  \phi \left( \nu \right) \phi \left( t \right)d\nu dt \\
     & = 1 + 2 \int_{-\infty}^{+\infty}  \tilde h^2\left( \beta + \nu \right) \phi \left( \nu \right) d\nu \int_{0}^{+\infty} \phi \left( t \right) dt \\
     \label{eq: integral0}
     &- 4  \int_{-\infty}^{+\infty}  \tilde h\left( \beta + \nu \right) \phi \left( \nu \right) d\nu \int_{0}^{+\infty}t \phi \left( t \right) dt \\
     \label{eq: integral1}
     &=1  +  \int_{-\infty}^{+\infty}  \tilde h^2\left( \beta + \nu \right) \phi \left( \nu \right) d\nu  \\
     \label{eq: integral2}
     &- \frac{4}{\sqrt{2 \pi}}\int_{-\infty}^{+\infty}  \tilde h\left( \beta + \nu \right) \phi \left( \nu \right) d\nu 
\end{align}
where the primitive for the leftmost integral in \eqref{eq: integral0} can be found in \cite[Identity 10 ]{Owen_1980}.

Since $ \tilde h\left( \beta + \nu \right) = \sqrt{2 / \pi} \tanh\left[\beta \left(\beta + \nu \right) \right]$, the integrals in \eqref{eq: integral1} and \eqref{eq: integral2} are not trivial to develop further. To promote tractability, we adopt a reasonable approximation through the error function $\operatorname{erf}(\xi) = 2 \Phi(\xi\sqrt{2}) - 1$: $\tanh(\xi) \simeq \operatorname{erf}\left( \xi \sqrt{\pi} / 2 \right)$. This is equivalent to using an $\operatorname{erf}$-based decoder instead of the optimal MMSE decoder, and hence we upper-bound the MMSE with 
\begin{align}
&\mathrm{MSE}\left( \tilde x, \breve{x}\right)  \\
\label{eq: integral3}
&= 1 + \frac{2}{\pi} \int_{-\infty}^{+\infty} \left\{ 2 \Phi\left[ \sqrt{\frac{\pi}{2}} \beta \left(\beta + \nu \right) \right] - 1 \right\}^2\phi \left( \nu \right) d\nu  \\
&- \frac{4}{\sqrt{2 \pi}} \sqrt{\frac{2}{\pi}} \int_{-\infty}^{+\infty} \left\{ 2 \Phi\left[ \sqrt{\frac{\pi}{2}} \beta \left(\beta + \nu \right) \right] - 1 \right\} \phi \left( \nu \right)  d\nu 
\end{align}
After developing the square in \eqref{eq: integral3}, and rearranging the terms
\begin{align}
&\mathrm{MSE}\left( \tilde x, \breve{x}\right) = 1 + \frac{6}{\pi} \int_{-\infty}^{+\infty} \phi \left( \nu \right)  d\nu\\
\label{eq: integral4}
&+  \frac{8}{\pi} \int_{-\infty}^{+\infty} \Phi^2\left[ \sqrt{\frac{\pi}{2}} \beta \left(\beta + \nu \right) \right]\phi \left( \nu \right) d\nu  \\
\label{eq: integral5}
&- \frac{16}{\pi}\int_{-\infty}^{+\infty} \Phi\left[\sqrt{\frac{\pi}{2}} \beta \left(\beta + \nu \right) \right] \phi \left( \nu \right)  d\nu
\end{align}
To solve \eqref{eq: integral4} and \eqref{eq: integral5}, we rely on \cite[Identities 20,010.4; 10,010.8]{Owen_1980}
\begin{align}
&\mathrm{MSE}\left( \tilde x, \breve{x}\right) = 1 + \frac{6}{\pi} - \frac{16}{\pi} \Phi \left( \frac{a}{\sqrt{1 + b^2}} \right)\\
&+ \frac{8}{\pi}\left[ \Phi \left( \frac{a}{\sqrt{1 + b^2}} \right) - 2T\left( \frac{a}{\sqrt{1 + b^2}}; \frac{1}{\sqrt{1 + 2b^2}}\right)\right]
\end{align}
where $a = \beta^2 \sqrt{\pi / 2}$, $b = \beta \sqrt{\pi / 2}$. Now using the identities $\Phi(-x) = 1 - \Phi(x)$, and $T(-h, c) = T(h, c)$ \cite{Owen_1956AMS} we finally get
\begin{align}
&\mathrm{MSE}\left( \tilde x, \breve{x}\right) = \frac{\pi-2}{\pi}\\
&+\frac{8}{\pi} \left[
Q\left(\frac{a}{\sqrt{1 + b^2}}  \right)
- 2T\left(\frac{-a}{\sqrt{1 + b^2}}, \;  \frac{1}{\sqrt{1 + 2b^2}}
\right)\right]  \\
&= \frac{\pi-2}{\pi} + \frac{8}{\pi}\Phi_\mathrm{SN}\left(-\frac{a}{\sqrt{1 + b^2}}; \frac{1}{\sqrt{1 + 2b^2}} \right)\\
&= \frac{\pi-2}{\pi} + \frac{8}{\pi}\Phi_{\mathrm{SN}}\!\left(
-\frac{\beta^2\sqrt{\pi}}{\sqrt{2+\pi\beta^2}};\;
\frac{1}{\sqrt{1+\pi\beta^2}}
\right)
\end{align}

\end{proof}

\section*{Appendix B}
\begin{proof}[Proof of Lemma~\ref{lemma: decoder ad}]
Similarly to the proof of Lemma~\ref{lemma: decoder}, given an encoder $\tilde g$, the MMSE decoder can be expressed as \cite{Akyol_2014TIT}[Theorem~3]
\begin{align}
      \tilde h_{\alpha}(\tilde w)&= \frac{\int_{-\infty}^{+\infty} \eta f_{\tilde Z}(\tilde{w} - \tilde g(\eta)) f_{\tilde X} (\eta)\, d\eta}{\int_{-\infty}^{+\infty} f_{\tilde Z}(\tilde w- \tilde g(t)) f_{\tilde X}(\eta) \, d\eta} \\
       \label{eq: dec ad 1}
       &= \frac{\int_{-t}^{+t}\phi\left(\tilde g(\eta) - \tilde w) \right) \phi\left(\eta\right) d\eta}{\int_{-t}^{+t}  \phi\left( \tilde g(\eta) - \tilde w\right) \phi\left( \eta \right) d\eta}\\
       \label{eq: dec ad 2}
       &= \frac{\int_{-t}^{+t}\phi\left(\alpha \eta - \tilde w \right) \phi\left(\eta\right) d\eta}{\int_{-t}^{+t}  \phi\left( \alpha \eta - \tilde w\right) \phi\left( \eta \right) d\eta} = \frac{\left[ N \right]^{+t}_{-t}}{\left[ D \right]^{+t}_{-t}}
\end{align}
The primitives $N$ and $D$ are given by \eqref{eq: primitive N} and \eqref{eq: primitive D} by replacing $\nu_\pm $ with $- \tilde w$:
\begin{align}
N&= \frac{-1}{\varsigma^2}\phi\left( \frac{ \tilde w}{\varsigma}\right) \left[\phi\left( \eta \varsigma - \frac{\alpha \tilde w}{\varsigma} \right )
- \frac{\alpha \tilde w }{\varsigma}\Phi\left( \eta \varsigma - \frac{\alpha \tilde  w }{\varsigma} \right) \right]\\
D  &= \frac{1}{\varsigma}\phi\left( \frac{\tilde w}{\varsigma}\right)\Phi\left( \eta \varsigma - \frac{\alpha \tilde w }{\varsigma} \right).
\end{align}
with $\varsigma = \sqrt{1 + \alpha^2}$.
\begin{align}
\tilde h_{\alpha}(\tilde w)&=\frac{\left[ N \right]^{+t}_{-t}}{\left[ D \right]^{+t}_{-t}}\\
    &=-\frac{1}{\varsigma}\frac{\phi\left( t \varsigma - \alpha \tilde w / \varsigma \right )
- \alpha \tilde w \Phi\left( t \varsigma - \alpha \tilde w / \varsigma  \right) / \varsigma  }{\Phi\left( t\varsigma - \alpha \tilde w / \varsigma  \right) - \Phi\left( -t \varsigma - \alpha \tilde w / \varsigma  \right)} \\
&+ \frac{1}{\varsigma}\frac{\phi\left( -t \varsigma - \alpha \tilde w / \varsigma  \right )
- \alpha \tilde w \Phi\left(- t \varsigma - \alpha \tilde w / \varsigma  \right) / \varsigma  }{\Phi\left( t\varsigma - \alpha \tilde w / \varsigma \right) - \Phi\left( -t \varsigma - \alpha \tilde w / \varsigma  \right)} \\
&= -\frac{1}{\varsigma} \frac{\phi\left(\alpha \tilde w / \varsigma - t \varsigma  \right ) - \phi\left(\alpha \tilde w / \varsigma +  t \varsigma  \right )}{Q\left( \alpha \tilde w / \varsigma - t\varsigma \right) - Q\left( \alpha \tilde w / \varsigma + t\varsigma  \right)} \\
&+ \frac{\alpha \tilde w}{\varsigma^2} \frac{Q\left( \alpha \tilde w / \varsigma - t\varsigma \right) - Q\left( \alpha \tilde w / \varsigma + t\varsigma \right)}{Q\left( \alpha \tilde w / \varsigma - t\varsigma \right) - Q\left( \alpha \tilde w / \varsigma + t\varsigma  \right)} \\
&= \frac{\alpha \tilde w}{\varsigma^2} - \frac{1}{\varsigma} \frac{\phi\left(\alpha \tilde w / \varsigma - t \varsigma  \right ) - \phi\left(\alpha \tilde w / \varsigma +  t \varsigma  \right )}{Q\left( \alpha \tilde w / \varsigma - t\varsigma \right) - Q\left( \alpha \tilde w / \varsigma + t\varsigma  \right)} 
\end{align}
Since $\phi(\pm \infty) = 0$, $Q(-\infty) = 1$ and $Q(-\infty) = 0$, for $t=+\infty$ and $- t= -\infty$, as expected, we obtain the MMSE decoder for the Gaussian signal
\begin{align}
\tilde h_{\alpha}(\tilde w) = \frac{\alpha \tilde w}{\varsigma^2} = \frac{\alpha }{1 + \alpha^2} \tilde w.
\end{align}

\end{proof}

\begin{proof}[Proof of the power constraint for anomaly detection]

A constraint on $\mathbb{E}_X [ g^2(x)] \leq P$, in terms of $ \tilde g( \tilde x)$ becomes
\begin{align}
& \mathbb{E} \left[ \tilde g^2(\xi) \right] = \int_{-\infty}^{\infty} \tilde g^2\left( \eta \right) f_{\tilde X}(\eta) d \eta \\
&= \int_{-t}^{t} \tilde g^2\left( \eta \right) \phi(\eta) d\eta + \int_{|\eta| > t}  \tilde g^2 \left(  \eta \right)\phi(\eta) d\eta \\
&= \theta \int_{-t}^{t}  \tilde g^2\left( \eta \right) \frac{\phi(\eta)}{\theta} d\eta + \left(1-\theta \right)\int_{|\eta| > t} \tilde g^2 \left(  \eta \right) \frac{\phi(\eta)}{1-\theta} d\eta \\
&= \theta \operatorname{SNR}^\ok + \left( 1 - \theta  \right) \operatorname{SNR}^\ko \leq \operatorname{SNR}
\end{align}
Using the definition of $\tilde g(\xi)$, we can express the two terms as
\begin{align}
\operatorname{SNR}^\ok &= 
\label{eq: integral power ok}
\int_{-t}^{+t}\alpha^2 \eta^2  \frac{\phi(\eta)}{\theta} d\eta = \frac{\alpha^2 }{\theta}\int_{-t}^{+t}\eta^2  \phi(\eta) d\eta \\
& = \frac{\alpha^2 }{\theta} \left[ \Phi(\eta) - \eta \phi(\eta)\right]^{+t}_{-t}\\
& = \frac{\alpha^2 }{\theta} \left[\Phi(t) - t \phi(t) - \Phi(-t) - t \phi(-t) \right] \\
&=  \frac{ \alpha^2 }{\theta} \left[1 - 2 Q(t) - 2 t \phi(t) \right] = \alpha^2 \left[ 1 - \frac{2t}{\theta} \phi \left( t \right)\right] 
\end{align}
where the primitive for the integral in \eqref{eq: integral power ok} can be found in \cite[Identity 12]{Owen_1980} and, in the last equality,  we exploit the definition of $\theta$.

While,
\begin{align}
\operatorname{SNR}^\ko &= 
\int_{-\infty}^{-t}\left(\beta \eta - \delta \right)^2  \frac{\phi(\eta)}{1-\theta} d\eta + \int_{+t}^{+\infty}\left(\beta \eta + \delta \right)^2  \frac{\phi(\eta)}{1-\theta} d\eta \\
\label{eq: integral power ok2}
&= \frac{2}{1-\theta}\int_{+t}^{+\infty}\left(\beta \eta + \delta \right)^2  \phi(\eta) d\eta 
\end{align}
where in the last equality we exploit symmetry. Now, to solve the integral we use rely on \cite[Identities 11, 12]{Owen_1980} and the definition of $Q(\xi)$
\begin{align}
&\int_{+t}^{+\infty}\left(\beta \eta + \delta \right)^2  \phi(\eta) = \beta^2 \int_{+t}^{+\infty} \eta^2  \phi(\eta) d\eta \\
&+ 2 \beta \delta \int_{+t}^{+\infty} \eta  \phi(\eta) d\eta + \delta^2 \int_{+t}^{+\infty} \phi(\eta) d\eta\\
&=\beta^2 \left[ \Phi(\eta) - \eta \phi(\eta)\right]_{+t}^{+\infty}-2 \beta \delta \left[ \phi(\eta)\right]_{+t}^{+\infty} + \delta^2 Q(t) \\
&= \beta^2 \left[ 1 - \Phi(t) + t \phi\left(t\right)\right] + 2 \beta \delta \phi(t) + \delta^2 Q(t) 
\end{align} 
where in the last equality we used the fact that $+\infty \cdot \phi(+\infty) = 0$ since the exponential decades faster to $0$ then the linear term grows to infinity. Finally, after rearranging the terms, from \eqref{eq: integral power ok2} we get
\begin{align}
&\operatorname{SNR}^\ko = \frac{2}{1-\theta} \left[ \left( \beta^2 + \delta^2 \right) Q(t) + \beta \left( \beta t + 2 \delta \right) \phi(t) \right]
\end{align}

\end{proof}

\begin{proof}[Proof of Lemma~\ref{lemma: detector}]
The log-likelihood scoring function of $w$ is \cite{Marchioni_TSMC2024}
\begin{align}
    \label{eq: optimal score}
    S\left( w \right) &= -\log{ f_{W^\ok} \left( w \right) }
\end{align}
which requires the channel's output distribution knowledge.
The output of the channel can be expressed by marginalization
\begin{align}
    \label{eq: output distribution}
f_{W}(w) &= \int_{-\infty}^{+\infty}  f_{X, W}(\xi, w)\, d\xi  \\
&= \int_{-\infty}^{+\infty}  f_X(\xi)  f_{W | X} \left(w | \xi \right)\, d\xi\\
&= \int_{-\infty}^{+\infty}  f_{X}\left( \xi \right)  f_Z\left[w- g\left(\xi\right)\right]\, d\xi
\end{align}
For the normal signal, we have
\begin{align}
&f_{W^\ok}(w) = \int_{-\infty}^{+\infty}  f_{X^\ok}\left( \xi \right)  f_Z\left[w- g\left(\xi\right)\right]\, d\xi \\
&=\int_{-T}^{+T} \frac{1}{\theta \sigma_X}\phi\left( \frac{\xi}{\sigma_X} \right) \frac{1}{\sigma_Z}\phi\left[ \frac{w - g(\xi) }{\sigma_Z} \right]  \, d\xi \\
\label{eq: integral output distribution}
&= \frac{1}{\theta \sigma_Z}\int_{-t}^{+t}\phi\left[ \frac{w}{\sigma_Z} - \tilde g(\eta) \right] \phi\left( \eta\right) \, d \eta \\
&= \frac{1}{\sigma_Z}\int_{-t}^{+t} f_{\tilde Z}\left[ \tilde w - \tilde g(\eta) \right] f_{\tilde X^\ok}\left( \eta\right) \, d \eta = \frac{1}{\sigma_Z} f_{\tilde W^\ok}(\tilde w) 
\end{align}
where to obtain \eqref{eq: integral output distribution}, we set $\xi = \sigma_X \eta$, and $t=T/\sigma_X$.
Now, since for the normal signal $\tilde g(\xi) = \alpha \xi$, \eqref{eq: integral output distribution}, up to a constant, matches the denominator of \eqref{eq: dec ad 2}, so that
\begin{align}
&f_{\tilde W^\ok}(\tilde w)  = \frac{1}{\theta}\left[ D \right]^{+t}_{-t} =\frac{1}{\theta}\left[\frac{1}{\varsigma}\phi\left( \frac{\tilde w}{\varsigma}\right)\Phi\left( \eta \varsigma - \frac{\alpha \tilde w }{\varsigma} \right)\right]^{+t}_{-t} \\
\label{eq: channels distribution final}
&= \frac{1}{\theta \varsigma }\phi\left( \frac{ \tilde w}{\varsigma}\right) \left[ Q\left( \frac{\alpha \tilde w}{\varsigma} - t\varsigma \right) - Q\left( \frac{\alpha \tilde w}{\varsigma} + t\varsigma  \right) \right]
\end{align}
with $\varsigma = \sqrt{1 + \alpha^2}$.

Note that, for $t=+\infty$, $Q(-\infty) = 1$, $Q(+\infty) = 0$, $\theta = 1$ and, $f_{\tilde W^\ok}(\tilde w) = \phi( \tilde w / \varsigma ) / \varsigma $, i.e., $\tilde W^{\ok} \sim \mathcal{N}\!\left(0,\,1+\alpha^2\right)$, as expected from $\tilde{w} = \alpha \tilde x + \tilde z$ where $\tilde x,\tilde  z \sim \mathcal{N}(0,1)$ and independent.

Now, with \eqref{eq: channels distribution final}, the optimal score defined in \eqref{eq: optimal score} can be written as
\begin{align}
 S\left( w \right) &= -\log{ f_{\tilde W^\ok} \left( \frac{w}{\sigma_Z} \right) } + \log{\sigma_Z} =  S\left( \tilde w \right) + \log{\sigma_Z}.
\end{align}
For a rank-based detector, the constant term is not relevant, so that $S\left( \tilde w \right)$ is equivalent to $ S\left( w \right) $. Moreover,
\begin{align}
S\left( \tilde w \right) & =\log{\sqrt{2 \pi}\theta \varsigma } + \frac{\tilde w^2}{2 \varsigma^2}\\
&- \log\left[Q\left( \frac{\alpha \tilde w}{\varsigma} - t\varsigma \right) - Q\left( \frac{\alpha \tilde w}{\varsigma} + t\varsigma  \right)\right]
\end{align}
is equivalent to
\begin{equation}
    \bar S\left( \tilde w \right) = \frac{\tilde w^2}{2\varsigma^2} - \log[q(\tilde w)]
\end{equation}
where
\begin{align}
\label{eq: q functional}
q\left( \tilde w\right) &= Q\left( \frac{\alpha \tilde w}{\varsigma} - t\varsigma \right) - Q\left( \frac{\alpha \tilde w}{\varsigma} + t\varsigma  \right) \\
\label{eq: q functional 2}
 &= Q\left( \frac{\alpha \tilde w}{\varsigma} - t\varsigma \right) + Q\left( -\frac{\alpha \tilde w}{\varsigma} - t\varsigma  \right) -1 
\end{align}
due to the fact that $Q(-\xi) = 1 - Q(\xi)$. This shows that $q( \tilde w )$ is an even function, i.e., $q\left( \tilde w\right) = q\left( -\tilde w\right)$. The first term in $\bar S( \tilde w )$ is even, the logarithm of an even function is even and so is the difference between two even functions. Therefore, $\bar S( \tilde w )$ is even and depends only on the absolute value of $ \tilde w$.

We now show that $\bar S\left( \tilde w \right)$ is strictly increasing on $\tilde w \geq 0$ by first computing
\begin{equation}
    \bar S'\left( \tilde w \right) = \frac{\tilde w}{\varsigma^2} - \frac{q'(\tilde w) }{q(\tilde w) }
\end{equation}
where, since $Q'(\xi) = -\phi(\xi)$, the derivative of $q(\tilde w)$ in 
\eqref{eq: q functional} is
\begin{equation}
    q'(\tilde w) = -\frac{\alpha}{\varsigma} \left[ \phi\left( \frac{\alpha \tilde w}{\varsigma} - t\varsigma \right) - \phi \left( \frac{\alpha \tilde w}{\varsigma} + t\varsigma  \right) \right].
\end{equation}
Given that for $ \tilde w > 0$
\begin{equation}
\left| \frac{\alpha \tilde w}{\varsigma} - t\varsigma \right| < \left| \frac{\alpha \tilde w}{\varsigma} + t\varsigma \right|
\end{equation}
and $\phi \left( \xi \right)$ is strictly decreasing in $|\xi|$,
\begin{equation}
\phi\left( \frac{\alpha \tilde w}{\varsigma} - t\varsigma \right) > \phi \left( \frac{\alpha \tilde w}{\varsigma} + t\varsigma  \right),
\end{equation}
therefore $q'(\tilde w) < 0$.

Finally, since the density must be positive, according to \eqref{eq: channels distribution final}, $q(\tilde w)$ is positive, hence $q'(\tilde w)/q(\tilde w)$ is overall negative. So for $\tilde w > 0$,  $\bar S'\left( w \right)$ is positive and therefore $\bar S\left(\tilde w \right)$ is a strictly increasing function of $\tilde w > 0$ and $|\tilde w|$.

\end{proof}

\begin{proof}[Proof of Lemma~\ref{lemma: detector performance}]
Assuming $\tilde x^\ko \sim f^\tails_{\tilde X^{\ko}}(\xi)$ the risk of $\gamma(\tilde w)$ can be expressed as
\begin{align}
    &\mathcal{R}(\gamma) = \Pr\left\{\gamma(\tilde W) \neq C \right\}=  \Pr\left\{ \left| \tilde W \right| > \psi \wedge \left| \tilde  X \right| <  t \right\}\\
    &+  \Pr\left\{\left| \tilde  W \right| < \psi \wedge \left| \tilde  X \right| > t \right\}=  \theta \Pr\left\{ \left| \tilde W \right| > \psi \mid \left| \tilde  X \right| <  t\right\}\\
    &+ \left( 1 - \theta \right) \Pr\left\{\left| \tilde  W \right| < \psi \mid \left| \tilde  X \right| > t \right\}\\
    &= \theta \operatorname{FPR} \left( \psi \right) + \left( 1 - \theta \right) \operatorname{FNR}\left( \psi \right)
\end{align}

\begin{align}
    &\operatorname{FPR} \left( \psi \right) = \int_{-t}^{+t} \Pr\left\{ \left| \tilde W \right| > \psi \mid \tilde X = \eta \right\} f_{\tilde X^\ok}(\eta) \, d\eta \\
    &= \frac{1}{\theta}\int_{-t}^{+t}   \Pr\left\{ \left| \tilde g(\eta) + \tilde Z \right| > \psi \right\} \phi(\eta) \, d\eta\\
    &= \frac{1}{\theta}\int_{-t}^{+t}   \Pr\left\{  \tilde g(\eta) + \tilde Z  > \psi \right\}  \phi(\eta) \, d\eta\\
    &+ \frac{1}{\theta}\int_{-t}^{+t}    \Pr\left\{  \tilde g(\eta) + \tilde Z  < - \psi \right\} \phi(\eta) \, d\eta\\
    &= \frac{1}{\theta}\int_{-t}^{+t}   \Pr\left\{ \tilde Z  > \psi - \alpha \eta  \right\}  \phi(\eta) \, d\eta\\
    &+ \frac{1}{\theta}\int_{-t}^{+t}    \Pr\left\{   \tilde Z  < - \psi -\alpha \eta \right\} \phi(\eta) \, d\eta \\
    &= \frac{1}{\theta}\int_{-t}^{+t}  \left[ 1 - \Phi\left(\psi - \alpha \eta  \right)\right] \phi(\eta) \, d\eta\\
    &+ \frac{1}{\theta}\int_{-t}^{+t}    \Phi\left(- \psi - \alpha \eta  \right)  \phi(\eta) \, d\eta \\
    &=\frac{1}{\theta}\int_{-t}^{+t} \left[\Phi\left(- \psi + \alpha \eta\right)  +  \Phi\left(- \psi - \alpha \eta  \right)\right]  \phi(\eta) \, d\eta
\end{align}
Now, according to \cite[Identity 10,010.4]{Owen_1980}
\begin{align}
    &\operatorname{FPR} \left( \psi \right) \\
    &= \frac{1}{\theta}\int^{-\frac{\psi}{\varsigma}}_{-\infty} \left[ \Phi \left( t \varsigma + \alpha \xi \right) - \Phi \left( -t \varsigma + \alpha \xi \right) \right]    \phi(\xi) d\xi \\
    \label{eq: fpr 1}
    &+ \frac{1}{\theta}\int^{-\frac{\psi}{\varsigma}}_{-\infty} \left[ \Phi \left( t \varsigma - \alpha \xi \right) - \Phi \left( -t \varsigma - \alpha \xi \right) \right]    \phi(\xi) d\xi \\
    \label{eq: fpr 2}
    &= \frac{2}{\theta} \int^{-\frac{\psi}{\varsigma}}_{-\infty}  \left[ \Phi \left( t \varsigma + \alpha \xi \right) - \Phi \left( -t \varsigma + \alpha \xi \right) \right]    \phi(\xi) d\xi
\end{align}
where the last equality uses $\Phi(-\xi) = 1 - \Phi(\xi)$ to show the two integrands are identical.
If we set $\alpha = - \rho / \sqrt{1-\rho^2} $, $t \varsigma = t / \sqrt{1-\rho^2}$,
\begin{align}
    &\operatorname{FPR} \left( \psi \right) = \frac{2}{\theta} \int^{-\frac{\psi}{\varsigma}}_{-\infty} \Phi \left( \frac{t - \rho \xi}{\sqrt{1-\rho^2}} \right)   \phi(\xi) d\xi\\
    &-\frac{2}{\theta} \int^{-\frac{\psi}{\varsigma}}_{-\infty} \Phi \left( \frac{-t - \rho \xi}{\sqrt{1-\rho^2}} \right)   \phi(\xi) d\xi \\
    \label{eq: fpr3}
    &= \frac{2}{\theta}\left[\Phi_2 \left(-\frac{\psi}{\varsigma}, t; \rho\right) - \Phi_2 \left(-\frac{\psi}{\varsigma}, -t; \rho\right) \right] \\
    &= \frac{2}{\theta}\left[\Phi_2 \left(-\frac{\psi}{\varsigma}, t; -\frac{\alpha}{\varsigma} \right) - \Phi_2 \left(-\frac{\psi}{\varsigma}, -t; -\frac{\alpha}{\varsigma}\right) \right]
\end{align}
where the expression \eqref{eq: fpr3} is given by \cite[Identity 10,010.2]{Owen_1980} with
\begin{align}
    &\Phi_2(h, k; \rho) \\
    &= \frac{1}{2\pi \sqrt{1-\rho^2}}\int_{-\infty}^{k}\int_{-\infty}^{h}\exp\left[-\left( \frac{\nu^2 - 2\rho \nu \eta + \eta^2}{2 (1-\rho^2)} \right)\right]d\nu d\eta
\end{align}
the CDF of a bi-variate Gaussian distribution.

Similarly,
\begin{align}
    &\operatorname{FNR} \left( \psi \right) = \int_{|\eta| > t} \Pr\left\{ \left|\tilde  W \right| < \psi \mid X = \eta \right\} f^\tails_{\tilde X^\ko}(\eta) \, d\eta\\
    &= \frac{1}{1-\theta}\int_{|\eta| > t} \Pr\left\{ \left| \tilde g(\eta) + \tilde Z \right| < \psi \right\}  \phi(\eta) \, d\eta \\
    &= \frac{1}{1-\theta} \int_{-\infty}^{-t} \Pr\left\{ \left| \beta \eta - \delta  + \tilde Z \right| < \psi \right\} \phi(\eta) \, d\eta\\
    &+ \frac{1}{1-\theta} \int_{+t}^{+\infty} \Pr\left\{ \left| \beta \eta + \delta  + \tilde Z\right| < \psi \right\} \phi(\eta) \, d\eta
\end{align}
Exploiting symmetry,
\begin{align}
    &\operatorname{FNR} \left( \psi \right) = \frac{2}{1-\theta} \int_{-\infty}^{-t} \Pr\left\{ \left| \beta \eta - \delta  + \tilde Z \right| < \psi \right\} \phi(\eta) \, d\eta \\
    &= \frac{2}{1-\theta} \int_{-\infty}^{-t} \Pr\left\{ \tilde Z  < \psi - \beta \eta + \delta \right\} \phi(\eta) \, d\eta \\
    &- \frac{2}{1-\theta}  \int_{-\infty}^{-t} \Pr\left\{ \tilde Z  < -\psi - \beta \eta + \delta \right\} \phi(\eta) \, d\eta \\
    \label{eq: fnr integral 1}
    &= \frac{2}{1-\theta}  \int_{-\infty}^{-t} \Phi\left(\delta  + \psi - \beta \eta\right) \phi(\eta) \, d\eta \\
    \label{eq: fnr integral 2}
    &- \frac{2}{1-\theta} \int_{-\infty}^{-t}  \Phi \left( \delta - \psi  - \beta \eta \right)  \phi(\eta) \, d\eta
\end{align}
If we set $\beta = \omega / \sqrt{1-\omega^2} $, and $\kappa_\pm = (\delta \pm \psi ) / \sqrt{1-\omega^2}$, according to \cite[Identity 10,010.2]{Owen_1980}
\begin{align}
    &\operatorname{FNR} \left( \psi \right) = \frac{2}{1-\theta}  \int_{-\infty}^{-t} \Phi\left(\frac{\kappa_+ - \omega \eta}{\sqrt{1-\omega^2}}\right) \phi(\eta) \, d\eta \\
    &-\frac{2}{1-\theta}  \int_{-\infty}^{-t} \Phi\left(\frac{\kappa_-  -\omega \eta}{\sqrt{1-\omega^2}}\right) \phi(\eta) \, d\eta \\
    &= \frac{2}{1-\theta}\left[\Phi_2 \left(-t, \kappa_+; \omega\right) - \Phi_2 \left(-t, \kappa_-; \omega\right) \right] \\
    &= \frac{2}{1-\theta}\left[\Phi_2 \left(-t, \frac{\delta + \psi}{\vartheta}; -\frac{\beta}{\vartheta} \right) - \Phi_2 \left(-t, \frac{\delta - \psi}{\vartheta}; -\frac{\beta}{\vartheta}\right) \right]\\
    &= \frac{2}{1-\theta}\left[\Phi_2 \left(\frac{\delta + \psi}{\vartheta}, -t; -\frac{\beta}{\vartheta} \right) - \Phi_2 \left(\frac{\delta - \psi}{\vartheta}, -t; -\frac{\beta}{\vartheta}\right) \right]  
\end{align}
where $\vartheta = \sqrt{1 + \beta^2}$ and the last equality is due to symmetry of $\Phi_2(h, k; \rho)$.
\end{proof}

\begin{proof}[Proof of FNR for $\beta=0$]
When $\beta=0$, \eqref{eq: fnr integral 1} and \eqref{eq: fnr integral 2} become
\begin{align}
    &\operatorname{FNR} \left( \psi \right) = \frac{2}{1-\theta}  \int_{-\infty}^{-t} \Phi\left(\delta  + \psi \right) \phi(\eta) \, d\eta \\
    &- \frac{2}{1-\theta} \int_{-\infty}^{-t}  \Phi \left( \delta - \psi\right)  \phi(\eta) \, d\eta \\
    & = \left[\Phi\left(\delta  + \psi \right)  - \Phi \left( \delta - \psi\right)\right] \frac{2}{1-\theta} \int_{-\infty}^{-t} \phi(\eta) \, d\eta \\
    \label{eq: risk ad beta=0}
    &= \Phi\left(\delta  + \psi \right)- \Phi \left( \delta - \psi\right) = Q\left( \delta - \psi\right) - Q\left(\delta  + \psi \right) 
\end{align}
This result is immediately generalizable for any symmetric $f_{\tilde X^\ko}(\xi)$. 
\end{proof}

\begin{proof}[Proof of FNR for the uniform anomaly]
Assuming $f^\unknown_{\tilde X^\ko}(\xi) = f^\uniform_{\tilde X^\ko}(\xi) = 1 / [2(m - t)] \, \mathbf{1}_{[-m, -t] \cup [t, m]}(\xi)
$ with $m > t$ and following a similar rationale that led to \eqref{eq: fnr integral 1} and \eqref{eq: fnr integral 2} 
\begin{align}
    &\operatorname{FNR}^\uniform \left( \psi \right) = 2\int_{-\infty}^{-t} \Phi\left(\delta  + \psi - \beta \eta\right) f^\uniform_{\tilde{X}^{\ko}}(\eta)\, d\eta \\
    &- 2 \int_{-\infty}^{-t}  \Phi \left( \delta - \psi  - \beta \eta \right)  f^\uniform_{\tilde{X}^{\ko}}(\eta) \, d\eta \\
    &= \frac{1}{m-t}\int_{-m}^{-t} \left[\Phi\left(\delta  + \psi - \beta \eta\right) - \Phi \left( \delta - \psi  - \beta \eta \right) \right]d\eta  
\end{align} 
Now using \cite[Identity 10,000]{Owen_1980} with $\kappa_\pm = \delta  \pm \psi $
\begin{align}
& \operatorname{FNR}^\uniform \left( \psi \right) =\\
&- \frac{1}{m-t}\left[\left(\frac{\kappa_+ - \beta \eta}{\beta}\right) \Phi \left( \kappa_+ - \beta \eta \right) + \frac{1}{\beta} \phi( \kappa_+ - \beta \eta) \right]^{-t}_{-m} \\
&+ \frac{1}{m-t} \left[\left(\frac{\kappa_- - \beta \eta}{\beta}\right) \Phi \left( \kappa_- - \beta \eta \right) + \frac{1}{\beta} \phi( \kappa_- - \beta \eta)\right]^{-t}_{-m}
\end{align} 
If we define $G ( \xi ) =  \xi  \Phi (\xi) + \phi( \xi) $,
\begin{align}
& \operatorname{FNR}^\uniform \left( \psi \right) = -\frac{1}{\beta\left(m-t \right)} \left[ G(\kappa_+ - \beta \eta) \right]^{-t}_{-m}\\
&+ \frac{1}{\beta\left(m-t \right)}\left[ G(\kappa_- - \beta \eta)\right]^{-t}_{-m}\\
& = \frac{1}{\beta\left(m-t \right)} \left[- G(\kappa_+ + \beta t) + G(\kappa_+ + \beta m)\right]\\
&+ \frac{1}{\beta\left(m-t \right)} \left[G(\kappa_- + \beta t) - G(\kappa_- + \beta m) \right]
\end{align}

\end{proof}

\begin{proof}[Proof of the expression for the risk stationarity condition]
Given the risk expression
    \begin{align}
    &\mathcal{R}(\gamma_{\alpha, \beta, \delta}) 
    = 2\left[\Phi_2 \left(-\frac{\psi}{\varsigma}, t; -\frac{\alpha}{\varsigma} \right) - \Phi_2 \left(-\frac{\psi}{\varsigma}, -t; -\frac{\alpha}{\varsigma}\right) \right]\\
    &+2\left[\Phi_2 \left(\frac{\kappa_+}{\vartheta}, -t; \frac{-\beta}{\vartheta} \right) - \Phi_2 \left(\frac{\kappa_-}{\vartheta}, -t ; \frac{-\beta}{\vartheta}\right) \right]
    \end{align}
    with $\varsigma = \sqrt{1 + \alpha^2}$, $\vartheta = \sqrt{1 + \beta^2}$ and $\kappa_\pm = \delta \pm \psi$,
to compute the $\partial / \partial \psi[\mathcal{R}(\gamma_{\alpha, \beta, \delta})]$ we first note that, according to \cite[Identity 10,010.2]{Owen_1980}, the CDF can be written as
\begin{equation}
    \label{eq: CDF bivariate normal}
    \Phi_2(h, k; \rho) = \int^{h}_{-\infty} \Phi \left( \frac{k - \rho \xi}{\sqrt{1-\rho^2}} \right)   \phi(\xi) d\xi.
\end{equation}
The derivative of \eqref{eq: CDF bivariate normal} w.r.t. to the first argument now can be computed as
\begin{align}
    &\frac{\partial}{\partial h}\Phi_2(h, k; \rho)=
    \frac{\partial}{\partial h}\lim_{l \to -\infty}\int^{h}_{l} \Phi \left( \frac{k - \rho \xi}{\sqrt{1-\rho^2}} \right)   \phi(\xi) d\xi\\
    &=\lim_{l \to -\infty} \frac{\partial}{\partial h} \int^{h}_{l} \Phi \left( \frac{k - \rho \xi}{\sqrt{1-\rho^2}} \right)   \phi(\xi) d\xi\\
    &=\phi(h)\Phi \left(\frac{k - \rho h}{\sqrt{1-\rho^2}}\right).
\end{align}
where the last equality is due to the fundamental theorem of calculus.
With this, the derivative of the risk w.r.t. the detection threshold $\psi$ becomes
\begin{align}
&\frac{\partial}{\partial \psi}\mathcal{R}(\gamma_{\alpha, \beta, \delta}) = -\frac{2}{\varsigma} \phi\left(-\frac{\psi}{\varsigma}\right)  \Phi \left(\frac{t + \rho \psi /  \varsigma}{\sqrt{1-\rho^2}}\right)\\
&+\frac{2}{\varsigma} \phi\left(-\frac{\psi}{\varsigma}\right) \Phi \left(\frac{-t + \rho \psi /  \varsigma}{\sqrt{1-\rho^2}}\right) \\
&+ \frac{2}{\vartheta} \phi \left( \frac{\delta + \psi}{\vartheta}\right) \Phi\left[ \frac{-t + \omega (\delta + \psi) / \vartheta}{\sqrt{1-\omega^2}} \right] \\
&+ \frac{2}{\vartheta} \phi \left( \frac{\delta - \psi}{\vartheta}\right) \Phi\left[ \frac{-t + \omega (\delta - \psi) / \vartheta}{\sqrt{1-\omega^2}} \right] \\
&=\frac{2}{\varsigma} \phi\left(\frac{\psi}{\varsigma}\right) \left[ \Phi \left(\frac{-t\varsigma^2 - \alpha \psi }{\varsigma} \right) - \Phi \left(\frac{t\varsigma^2 - \alpha \psi }{\varsigma} \right) \right]\\
&+ \frac{2}{\vartheta} \phi \left( \frac{\delta + \psi}{\vartheta}\right) \Phi\left[ \frac{-t \vartheta^2 + \beta (\delta + \psi)}{\vartheta} \right] \\
&+ \frac{2}{\vartheta} \phi \left( \frac{\delta - \psi}{\vartheta}\right) \Phi\left[ \frac{-t \vartheta^2 + \beta (\delta - \psi)}{\vartheta} \right]
\end{align}
where the previous definitions $\alpha = - \rho / \sqrt{1-\rho^2} $, $t \varsigma = t / \sqrt{1-\rho^2}$, $\vartheta = \sqrt{1 + \beta^2}$ and $\beta = \omega / \sqrt{1-\omega^2} $ have been used.
Rewriting in terms of $Q$,
\begin{align}
\label{eq: risk stationary}
&\frac{\partial}{\partial \psi}\mathcal{R}(\gamma_{\alpha, \beta, \delta}) \\
&= \frac{2}{\varsigma} \phi\left(\frac{\psi}{\varsigma}\right) \left[ Q\left(\frac{\alpha \psi + t\varsigma^2 }{\varsigma} \right) - Q \left(\frac{\alpha \psi  - t\varsigma^2}{\varsigma} \right) \right]\\
&+ \frac{2}{\vartheta} \phi \left( \frac{\delta + \psi}{\vartheta}\right) Q\left[ \frac{t \vartheta^2  - \beta (\delta + \psi)}{\vartheta} \right] \\
&+ \frac{2}{\vartheta} \phi \left( \frac{\delta - \psi}{\vartheta}\right) Q\left[ \frac{ t \vartheta^2 - \beta (\delta - \psi)}{\vartheta} \right]
\end{align}

\end{proof}

\begin{proof}[Proof of Proposition~\ref{prop: design ad}]
Assuming  $\epsilon=0$, and adopting sign-only anomaly encoding ($\beta = 0$) the power constraint in \eqref{eq: power AD} results in
  \begin{equation}
\alpha^2 \left[ \theta - 2 t  \phi \left(t \right)\right] + 2 \delta^2 Q\left( t \right)  \leq \operatorname{SNR}.
 \end{equation}

On the other hand, FNR is given by \eqref{eq: risk ad beta=0} and together with Lemma~\ref{lemma: detector performance}, risk constraint becomes:

\begin{align}
& 2 \Phi_2 \left(- \frac{\psi}{ \varsigma}, t; -\frac{\alpha}{\varsigma} \right) - 2 \Phi_2 \left( - \frac{\psi}{\varsigma}, -t; -\frac{\alpha}{\varsigma} \right)  \\
& + \left( 1 - \theta \right)\left[ Q( \delta -  \psi )  -  Q ( \delta + \psi  ) \right] \leq P_e.
\end{align}

Finally, with $\beta=0$, the risk stationary condition in \eqref{eq: risk stationary} can be written as
\begin{align}
&\frac{2}{\varsigma} \phi\left(\frac{\psi}{\varsigma}\right) \left[ Q\left(\frac{\alpha \psi + t\varsigma^2 }{\varsigma} \right) - Q \left(\frac{\alpha \psi  - t\varsigma^2}{\varsigma} \right) \right]\\
\label{eq: derivative risk beta=0}
&+ 2 Q\left(t \right) \left[ \phi \left( \delta + \psi\right)
+ \phi \left( \delta - \psi\right) \right] = 0
\end{align}
Now, we assume that constraints are active at the solution. This heuristic leads to the fact that the design can be performed by solving
\begin{equation}
\label{eq: sytem 3 equations}
\displaystyle
\begin{dcases}
    \alpha^2 \left[ \theta - 2 t  \phi \left(t \right)\right] + \delta^2 \left( 1- \theta \right)  = \operatorname{SNR}\\
    2 \left[\Phi_2 \left(-\frac{\psi}{\varsigma}, t; -\frac{\alpha}{\varsigma} \right) - \Phi_2 \left( -\frac{\psi}{\varsigma}, -t; -\frac{\alpha}{\varsigma} \right) \right] \\
     + \left( 1 - \theta \right)\left[ Q( \delta -  \psi )  -  Q ( \delta + \psi  ) \right] = P_e  \\
     \frac{2}{\varsigma}\phi\left(\frac{\psi}{\varsigma}\right) \left[ Q\left(\alpha \frac{\psi}{\varsigma}  + t\varsigma  \right) - Q \left(\alpha \frac{\psi}{\varsigma} - t\varsigma \right) \right]\\
+ \left( 1 - \theta \right) \left[ \phi \left( \delta + \psi\right)
+ \phi \left( \delta - \psi\right) \right] = 0 
\end{dcases}
\end{equation}
where for \eqref{eq: derivative risk beta=0} we have exploited the fact that $2 Q(t) = 1 - \theta$.

From the power equation, we have that
\begin{equation}
    \delta(\alpha) = \sqrt{\frac{\operatorname{SNR} - \alpha^2 \left[\theta - 2 t \phi\left( t \right) \right]}{1 - \theta}},
\end{equation}
which is feasible iff $\alpha^2 \le \operatorname{SNR}/[\theta-2t\,\phi(t)]$. We also consider $\delta > 0$ as with the detector defined by Lemma~\ref{lemma: detector}, the sign is irrelevant. Eliminating $\delta$ from~\eqref{eq: sytem 3 equations} yields the two-equation
nonlinear system in $(\alpha,\psi)$:
\begin{equation}
\begin{dcases}
    2 \left[\Phi_2 \left(- \frac{\psi}{\varsigma}, t; -\frac{\alpha}{\varsigma} \right) - \Phi_2 \left( - \frac{\psi}{\varsigma}, -t; -\frac{\alpha}{\varsigma} \right) \right] \\
     + \left( 1 - \theta \right)\left[ Q( \delta(\alpha) -  \psi )  -  Q ( \delta(\alpha) + \psi  ) \right] = P_e  \\
     \frac{2}{\varsigma}\phi\left(\frac{\psi}{\varsigma}\right) \left[ Q\left(\alpha \frac{\psi}{\varsigma}  + t\varsigma  \right) - Q \left(\alpha \frac{\psi}{\varsigma}  - t\varsigma \right) \right]\\
+ \left( 1 - \theta \right) \left[ \phi \left( \delta(\alpha) + \psi\right)
+ \phi \left( \delta(\alpha) - \psi\right) \right] = 0 
\end{dcases}
\end{equation}
The solution $(\alpha^\star,\psi^\star)$ gives $\delta^\star=\delta(\alpha^\star)$ and completes the design.

\end{proof}

\end{document}